\documentclass[11pt]{article}
\usepackage[T1]{fontenc}
\usepackage{amssymb, amsmath, amsthm, graphicx, subfigure, mathtools}
\usepackage[margin=1in]{geometry}
\mathtoolsset{showonlyrefs}
\usepackage{tikz}

\usepackage{float}
\usepackage[colorlinks=true, allcolors=blue]{hyperref}

\usepackage{xcolor}
\usepackage{braket}
\usepackage{esvect}
\usepackage{enumerate}
\usepackage{scalefnt}
\usepackage{mdframed}

\usepackage{algorithm}
\usepackage{algorithmicx}
\usepackage[noend]{algpseudocode}

\DeclareMathOperator*{\E}{\mathbb{E}}
\DeclareMathOperator*{\Var}{\mathrm{Var}}

\newcommand{\THRESH}{{\cal THRESH}}

\newcommand{\bone}{\mathbf{1}}

\newcommand{\br}{\mathbf{r}}
\newcommand{\bv}{\mathbf{v}}
\newcommand{\bu}{\mathbf{u}}

\newcommand{\cP}{\mathcal{P}}

\newcommand{\GE}{\;\ge\;}

\newcommand{\true}{\mathtt{true}}
\newcommand{\false}{\mathtt{false}}

\newcommand{\sat}{\mathtt{sat}}
\newcommand{\unsat}{\mathtt{unsat}}
\newcommand{\unassigned}{\ast}

\newcommand{\Cov}{\mathrm{Cov}}

\newcommand{\dd}{\mathrm{d}}
\newcommand{\ee}{\mathrm{e}}

\newcommand{\He}{\mathrm{He}}

\newcommand{\MAXCSP}{\mathrm{MAX\,\,CSP}}

\newcommand{\Val}{\mathrm{Val}}

\newcommand{\comp}{\mathsf{Completeness}}
\newcommand{\sound}{\mathsf{Soundness}}

\newcommand{\supp}{\mathrm{supp}}

\newcommand{\buperp}[2]{\bu_{#1}^{\perp\{#2\}}}
\newcommand{\bvperp}[2]{\bv_{#1}^{\perp\{#2\}}}

\newtheorem{theorem}{Theorem}[section]
\newtheorem{definition}[theorem]{Definition}

\newtheorem{corollary}[theorem]{Corollary}
\newtheorem{lemma}[theorem]{Lemma}

\newtheorem{remark}[theorem]{Remark}

\newtheorem*{conjecture*}{Conjecture}

\title{Threshold Rounding and Bounded-Degree Boolean MAX 2-CSP}

\author{Suprovat Ghoshal\thanks{Indiana University.} \and Neng Huang\thanks{University of Michigan.} \and Euiwoong Lee\thanks{University of Michigan. Supported in part by NSF award CCF-2236669.} \and Konstantin Makarychev\thanks{Northwestern University. Supported by NSF awards CCF-1955351 and EECS-2216970.} \and Yury Makarychev\thanks{Toyota Technological Institute at Chicago. Supported by NSF award EECS-2216899.}}

\begin{document}

\maketitle

\begin{abstract}
    We describe an $\widetilde{\Omega}(1/d^4)$-improvement over threshold rounding schemes for a broad class of Boolean MAX 2-CSP instances in which every variable appears in at most $d$ constraints. In the case of MAX 2-SAT, we improve the ratio further and obtain an $(\beta_\star + \widetilde{\Omega}(1/d^2))$-factor approximation algorithm for bounded-degree MAX 2-SAT instances, where $\beta_\star$ is the UGC-optimal approximation ratio for MAX 2-SAT achieved by the LLZ algorithm~\cite{lewin2002improved}. Our result generalizes an $(\alpha_{GW} + \widetilde{\Omega}(1/d^2))$-factor approximation algorithm for MAX CUT on graphs with degrees bounded by $d$, due to Hsieh and Kothari~\cite{hsieh2023approximating}. Together with the state-of-the-art approximability results for MAX DI-CUT and MAX 2-AND~\cite{brakensiek2023separating}, our result suggests that similar improvements exist for bounded-degree instances of these problems as well. 
\end{abstract}

\section{Introduction}

Maximum Constraint Satisfaction Problems (MAX CSPs) are a fundamental class of combinatorial problems that are studied extensively in theoretical computer science. The design of approximation algorithms for MAX CSPs has been a central topic of research that has seen substantial progress in the past decades. Some of the main highlights include the design of approximation algorithms for MAX CSPs~\cite{MM17CSP} and the development of a framework for hardness reductions to study the limits of approximation for CSPs~\cite{Khot05Guest,KhotUGCSurvey}. These developments have culminated in surprisingly tight characterizations of the optimal approximation factor of every MAX CSP~\cite{Rag08}, under the Unique Games Conjecture (UGC)~\cite{Khot02a}.

In this work, we aim to explore the approximability of MAX CSPs for structurally simpler instances, namely, when the instances are {\em sparse}. Specifically, we explore the question of whether one can improve on the optimal approximation guarantee of a MAX CSP when the degree of the underlying constraint graph is bounded. Formally, we are interested in the following:
\vspace{2pt}
\begin{gather*}
	\textnormal{\it Given a MAX CSP problem $P$, suppose its optimal approximation factor is $\alpha_P$. Then, if the degree} \\
	\textnormal{\it of the instance is bounded by some constant $d$, can we efficiently achieve $\alpha_P + \epsilon_d$ approximation}, \\
	\textnormal{\it where the improvement $\epsilon_d$ depends on the degree $d$?}
\end{gather*}
\vspace{2pt}
The above question has been studied for various specific MAX CSPs in the literature~\cite{BK99,HLZ04,FKL02,AKS09,BMORRSTVWW15,hsieh2023approximating}. Arguably, one of the most interesting lines of work among these involves the approximability of MAX CUT in this setting. It was first studied by Feige, Karpinski, and Langberg~\cite{FKL02}, who showed that the MAX CUT problem on graphs with maximum degree $d$ admits an efficient approximation guarantee of $\alpha_{GW} + {\Omega}(1/d^4)$, where $\alpha_{GW} \approx 0.87856$ denotes the optimal approximation factor for MAX CUT under the UGC~\cite{GW94,KKMO07}. This was then improved by Flor\'en~\cite{Floren16} up to an $\alpha_{GW} + \Omega(1/d^3)$ factor, and by Hsieh and Kothari~\cite{hsieh2023approximating} to a factor of $\alpha_{GW} + \widetilde{\Omega}(1/d^2)$. Another notable work in this context is that of~\cite{BMORRSTVWW15}, who showed that one can bypass the approximation resistance of MAX $k$-LIN and other related problems in bounded-degree instances.

Unfortunately, generalizing this phenomenon beyond the specific examples mentioned above has been challenging. In particular, such results are not known even for well-studied friends of the MAX CUT problem, such as MAX $2$-SAT and MAX $2$-AND. One of the main difficulties in obtaining similar guarantees for general MAX CSPs has been the following: most of the aforementioned improvements in the bounded-degree setting implicitly rely on tight structural characterizations of the hardest-to-round, i.e., integrality gap, instances of the problem\footnote{We elaborate the connection between integrality gaps and improvements in the bounded-degree setting in Section \ref{sec:gaps}.}. Such explicit characterizations are known only for a handful of examples (e.g., MAX CUT~\cite{feige2002optimality}, MAX $k$-LIN~\cite{Has01}, MAX $2$-SAT~\cite{lewin2002improved,brakensiek2024tight}), and even the question of tightly characterizing the gap instances for every Boolean $2$-CSP is a long-standing open question.

To get around this bottleneck, we instead consider improving the approximation guarantee of a family of algorithms in the degree-bounded setting. Specifically, we consider the family of $\THRESH^-$-rounding scheme-based algorithms described informally below\footnote{We refer the reader to Section \ref{sec:prelim} for a more formal description of the SDP relaxation and the $\THRESH^-$ rounding scheme.}:
\begin{figure}[ht!]
	\begin{mdframed}
		\begin{center}
			{\bf\underline{SDP + $\THRESH^-$-Rounding Algorithm}}
		\end{center}
		\begin{enumerate}
			\item Solve the canonical semi-definite programming (SDP) relaxation for the MAX CSP.\\[-20pt]
			\item Round the SDP solution using the $\THRESH^-$ family of rounding functions.
		\end{enumerate}
	\end{mdframed}
	\caption{The $\THRESH^-$-based Approximation Framework}
	\label{alg:thresh}
\end{figure}
\newline
The above framework was first introduced by Lewin, Livnat, and Zwick~\cite{lewin2002improved}, and it vastly generalizes the hyperplane rounding-based algorithm of Goemans and Williamson~\cite{GW94}. This algorithmic framework has been used successfully for obtaining optimal and near-optimal guarantees for a wide variety of MAX CSPs such as MAX $2$-SAT~\cite{lewin2002improved,brakensiek2024tight}, MAX $2$-AND~\cite{brakensiek2023separating}, MAX BISECTION~\cite{austrin2016better}, and several other related problems. Furthermore, Austrin showed that under a certain positivity condition, the above framework is optimal for approximating Boolean MAX $2$-CSPs assuming UGC~\cite{Aus10}. Guided by these considerations, we are therefore motivated to ask --  {\em Can we improve on the approximation guarantee of the $\THRESH^-$-based algorithm, when the degree of the input instance is bounded?} --  which is the main question explored by this work.

\subsection{Our Results}

In this work, we make progress on the above questions and give improved approximation guarantees for bounded-degree Boolean MAX $2$-CSP instances which satisfy a natural ``reasonable'' condition. The formal definition of the reasonable condition is a bit technical, and we instead attempt to give an informal description of it. Let $\mathcal{P}$ denote a set of Boolean predicates, and let MAX CSP($\mathcal{P}$) denote the MAX CSP problem where the predicates corresponding to the constraints belong to $\mathcal{P}$. Then loosely speaking, we say that the MAX CSP($\mathcal{P}$) problem is reasonable if for any hardest-to-round instance\footnote{An instance $I$ of MAX CSP($\cP$) is said to be hardest-to-round if the algorithm from Figure \ref{alg:thresh} achieves its worst approximation guarantee in $I$.} of MAX CSP($\cP$), the local geometric arrangement of the vector solution in the neighborhood of any vertex satisfies certain ``general position'' conditions. In particular, these conditions ensure that certain pairwise moments involving labels of variables in the neighborhood are non-constant under the rounding distribution (see Lemma \ref{lem:cov_lower_bound}). We refer the reader to Definitions \ref{def:aligned} and \ref{def:reasonable_dist} for a formal description of this condition.

We now state our main result, which gives improved approximation guarantees for reasonable CSPs.

\begin{theorem}[Informal version of Theorem \ref{thm:general}]   \label{thm:intro-1}
	Suppose $\mathcal{P}$ is a set of Boolean binary predicates such that the MAX CSP($\cP$) problem is reasonable. Let $\beta_{\cP}$ denote the approximation factor achieved by the algorithm in Figure \ref{alg:thresh} for MAX CSP($\cP$). Let $I$ be an instance of MAX CSP($\cP$) such that every variable appears in at most $d$ constraints. Then, there exists an algorithm that achieves an approximation factor of $\beta_{\cP} + \widetilde{\Omega}(d^{-4})$ on $I$, where the constant hidden in $\widetilde{\Omega}$ only depends on $\mathcal{P}$.
\end{theorem}

While the above theorem by itself just implies improvement on the approximation guarantee of $\THRESH^-$-based algorithms, it has stronger implications for MAX CSPs that satisfy certain additional properties. As mentioned above, Austrin~\cite{Aus10} showed that if $\cP$ satisfies the positivity condition, then the algorithm from Figure \ref{alg:thresh} obtains the optimal approximation guarantee assuming UGC. Therefore, for MAX CSP($\cP$) problems that satisfy the reasonable and positivity conditions, the above theorem actually implies improvement on the optimal approximation guarantee in the bounded-degree setting. 

Furthermore, by leveraging the structural insights from the recent work \cite{brakensiek2024tight}, we are able to derive a sharper version of the above theorem for the MAX $2$-SAT problem. We formally state this result in the next theorem.

\begin{theorem}  \label{thm:intro-2}
	Let $I$ be an instance of MAX $2$-SAT such that every variable appears in at most $d$ constraints. Then there is an efficient algorithm that can return a $\beta_\star + \Omega(d^{-2}\log^{-1}(d))$-approximation to $I$, where $\beta_\star \approx 0.9401$ is the optimal approximation guarantee for the MAX $2$-SAT problem.
\end{theorem}

We point out a couple of interesting features of the above theorem. Firstly, observe that the above theorem does not mention the `reasonable' condition. This is because the structural results from \cite{brakensiek2024tight} directly imply that the hardest-to-round instances of MAX $2$-SAT naturally satisfy the reasonable condition, and hence we are able to apply our framework to obtain an unconditional improvement for degree-bounded MAX $2$-SAT instances. Furthermore, by using the specific structure of the $2$-SAT predicate, we are able to improve the additive gain in the approximation factor from $\widetilde{\Omega}(d^{-4})$ to $\widetilde{\Omega}(d^{-2})$.

\subsection{On Integrality Gaps and Bounded-degree CSPs} 		\label{sec:gaps}

We close out this section by highlighting the connection between hard-to-round instances and improved approximation for bounded-degree CSPs. Most of the previous works in this setting adopt the approach of improving the performance of a base algorithm via a greedy post-processing step that relies on certain anti-concentration properties of the objective function at local neighborhoods (similar to Lemma \ref{lem:2sat_local_gain}). In particular, such an approach has to be capable of interpolating between the following extremal scenarios:
\begin{itemize}
	\item For instances that are far from being an integrality gap, the base algorithm by itself is capable of yielding improved approximation guarantees, 
	\item And for instances that are close to being an integrality gap, the post-processing procedure should be able to leverage the bounded-degree structure of the instance to improve on the worst-case approximation guarantee.
\end{itemize}
Now, note that under this framework, it stands to reason that the greedy post-processing needs to be more effective for instances that are closer to being an integrality gap, whereas it may be less useful for easier instances -- e.g., for instances which are perfectly solvable by the base algorithm, the procedure cannot improve the quality of the solution output by the base algorithm any further. Therefore, the analysis of the  post-processing step necessarily requires a structural understanding of the hardest instances for the base algorithm. For example, for the MAX CUT problem, the analyses of algorithms in \cite{FKL02,Floren16,hsieh2023approximating} use the well-known characterization that the edges that are hardest to round by hyperplane rounding are those for which the inner product of the corresponding SDP vectors is equal to a certain critical value. 

We conclude by noting that while it might be possible to obtain improvements without explicitly characterizing the hardest instances, such a result will require algorithms that go beyond the above two-phase approach.

\subsection{Organization of the Paper}

The rest of the paper is organized as follows. In Section~\ref{sec:prelim}, we formally define MAX CSP, the Basic SDP relaxation, the $\THRESH^-$ rounding family, and related concepts. In Section~\ref{sec:algo}, we present our algorithm and define some algorithmic quantities. In Section~\ref{sec:2-sat}, we analyze the algorithm in the case of MAX 2-SAT, proving Theorem~\ref{thm:intro-2}. In Section~\ref{sec:2csp}, we analyze our algorithm for general Boolean MAX 2-CSPs, where we state and prove the formal version of Theorem~\ref{thm:intro-1}. Finally, Appendix~\ref{app:hermite} collects some facts about Hermite polynomials, which are then used in Appendix~\ref{app:variance} to derive certain variance lower bounds needed in the analysis. 

\section{Preliminaries}         \label{sec:prelim}

For $n \in \mathbb{N}$, let $[n] \coloneqq \{1, 2, \ldots, n\}$. Let $\Phi, \varphi$ be the c.d.f. and p.d.f. of a standard Gaussian variable. Given an event $A$, we use $\bone\{A\}$ to denote the indicator random variable for $A$.
We use $\true, \false$ for the truth values of a Boolean variable. We may also represent the truth values using $\{0, 1\}$ or $\{-1, 1\}$: in the former case 1 stands for $\true$ and 0 stands for $\false$, while in the latter case $-1$ stands for $\true$ and $1$ stands for $\false$.

\subsection{MAX CSP and the Basic SDP Relaxation}

\begin{definition}
    Let $\mathcal{P}$ be a set of Boolean predicates. An instance $I = (V, \mathcal{C}, w)$ of $\MAXCSP(\mathcal{P})$ is defined by
    \begin{itemize}
        \item A finite set $V$ of variables.
        \item A finite set $\mathcal{C}$ of constraints where each constraint $C \in \mathcal{C}$ is obtained by applying some predicate $P$ to some variables from $V$. We use $C = (P, i_1, i_2, \ldots, i_k)$ to denote a constraint $C$ obtained by applying some $k$-ary predicate $P$ to variables $i_1, i_2, \ldots, i_k$.
        \item A weight function $w: \mathcal{C} \to \mathbb{R}^{\geq 0}$.
    \end{itemize}
    Given an assignment $A: V \to \{-1, 1\}$, we say that a constraint $C$ is satisfied if the corresponding predicate evaluates to 1, and we use $\Val(I, A)$ to denote the total weight of constraints satisfied by $A$. The objective is to find an assignment $A$ that maximizes $\Val(I, A)$.
\end{definition}

For simplicity, we will assume without loss of generality that $\sum_{C \in \mathcal{C}} w(C) = 1$. This can always be achieved by renormalizing, which does not meaningfully change the problem. In this case we may also think of $\mathcal{C}$ as a distribution over constraints, with probabilities specified by $w$. We will also assume $V = [n]$ for some $n \in \mathbb{N}$. 

In this paper we only consider the special case of MAX 2-CSP, namely the case where every Boolean predicate in $\mathcal{P}$ is of arity 2. For convenience, we encode truth assignments for variables using $\{-1,1\}$, while predicates take values in $\{0,1\}$, and thus each predicate has the form $P: \{-1, 1\}^2 \to \{0, 1\}$. Given a MAX 2-CSP instance $I = (V, \mathcal{C}, w)$, we can formulate its \emph{Basic SDP relaxation}~\cite{Rag08} as follows (cf.~\cite{brakensiek2024tight}).

\begin{alignat}{3}
&\text{maximize} &\qquad& \sum_{C=(P, i, j) \in \mathcal C}  w(C) \cdot \left( \hat{P}_\varnothing + \hat{P}_i\bv_0\cdot\bv_{i} +\hat{P}_j \bv_0\cdot\bv_{j} + \hat{P}_{i,j} \bv_{i}\cdot\bv_{j}\right)\\
&\text{subject to} &\qquad& \forall i \in \{0, 1, 2, \ldots, n\}\;,\ \ \ \qquad\,\, \bv_i \cdot \bv_i = 1\;,\\
& &\qquad& \forall C = (P, i, j) \in \mathcal{C}\;, \quad
\begin{array}{c}
(\bv_0 - \bv_i) \cdot (\bv_0 - \bv_j) \GE 0\;,\\
(\bv_0 + \bv_i) \cdot (\bv_0 - \bv_j) \GE 0\;,\\
(\bv_0 - \bv_i) \cdot (\bv_0 + \bv_j) \GE 0\;,\\
(\bv_0 + \bv_i) \cdot (\bv_0 + \bv_j) \GE 0\;.\\
\end{array}\label{eq:triangle_0ij}
\end{alignat}

Here $\hat{P}_\varnothing, \hat{P}_i, \hat{P}_j, \hat{P}_{i, j} \in \mathbb{R}$ are the \emph{Fourier coefficients} of $P$, namely, we have
\[
    P(x_i, x_j) = \hat{P}_\varnothing + \hat{P}_ix_i + \hat{P}_jx_j + \hat{P}_{i, j}x_ix_j.
\]
for every $x_i, x_j \in \{-1, 1\}$ (See~\cite{o2014analysis} for more on Fourier analysis of Boolean functions). The constraints in~\eqref{eq:triangle_0ij} are often referred to as \emph{triangle inequalities}. In fact, we may in addition impose triangle inequalities on all triples of vectors:
\begin{equation}
\forall i, j, k \in \{0, 1, \ldots, n\}\;, \quad
\begin{array}{c}
(\bv_i - \bv_j) \cdot (\bv_i - \bv_k) \GE 0\;,\\
(\bv_i + \bv_j) \cdot (\bv_i - \bv_k) \GE 0\;,\\
(\bv_i - \bv_j) \cdot (\bv_i + \bv_k) \GE 0\;,\\
(\bv_i + \bv_j) \cdot (\bv_i + \bv_k) \GE 0\;.\\
\end{array}\label{eq:triangle_ijk}
\end{equation}

\begin{definition}
    We say that unit vectors $\bv_0, \bv_1, \ldots, \bv_n$ are an SDP solution for some MAX 2-CSP instance over the variable set $[n]$, if they satisfy~\eqref{eq:triangle_0ij}. We say that they satisfy all triangle inequalities if they further satisfy~\eqref{eq:triangle_ijk}. For $i,j \in [n]$, we write
    \[
        b_i = \bv_0 \cdot \bv_i
        \qquad\text{and}\qquad
        b_{ij} = \bv_i \cdot \bv_j .
    \]
\end{definition}

We will often need to take projections for an SDP solution. For $i \in [n]$ and $S \subseteq \{0\} \cup [n]$, we use $\bvperp{i}{S}$ to denote the (renormalized) component in $\bv_i$ that is orthogonal to $\mathrm{span}\{\bv_j \mid j \in S\}$. In particular, $\bvperp{i}{0} = \frac{\bv_i - b_i \bv_0}{\sqrt{1 - b_i^2}}$ is the (renormalized) component in $\bv_i$ that is orthogonal to $\bv_0$ (in the case $|b_i| = 1$, we may define $\bvperp{i}{0}$ to be any unit vector that is orthogonal to all other vectors in the SDP solution, but this case won't show up in our paper). We will write $\bv_i^\perp$ for $\bvperp{i}{0}$ to avoid clutter. 

\subsection{Configurations}

To analyze SDP-based approximation algorithms, it is sometimes more convenient to work with the inner products between vectors in an SDP solution rather than the vectors themselves. This leads to the following definition.

\begin{definition}
    A \emph{configuration} is a tuple of the form $(b_i, b_j, b_{ij}, P)$ where $P$ is a Boolean predicate of arity 2, and $b_i, b_j, b_{ij} \in [-1, 1]$ are real numbers such that \begin{equation}\label{eq:valid}-1 + |b_i + b_j| \leq b_{ij} \leq 1 - |b_i - b_j|.\end{equation} We say that a tuple $(b_i, b_j, b_{ij})$ is \emph{valid} if it satisfies \eqref{eq:valid}. Given a configuration $\theta = (b_i, b_j, b_{ij}, P)$, we define $\rho(\theta) = \frac{b_{ij} - b_ib_j}{\sqrt{(1 - b_i^2)(1 - b_j^2)}}$ if $\sqrt{(1 - b_i^2)(1 - b_j^2)} \neq 0$ and $\rho(\theta) = 0$ otherwise.
\end{definition}

It can be easily verified that the inequality \eqref{eq:valid} corresponds exactly to the triangle inequalities in~\eqref{eq:triangle_0ij}, and therefore a tuple $(b_i, b_j, b_{ij})$ is valid if and only if it can be obtained as inner products between $\bv_0, \bv_i, \bv_j$ from an SDP solution. In this case, $\rho(\theta)$ is exactly the inner product between $\bv_i^\perp$ and $\bv_j^\perp$. 

\begin{definition}          \label{def:aligned}
    Given a MAX 2-CSP instance $I = (V, \mathcal{C}, w)$ and an SDP solution $\bv_0, \cdots, \bv_n$, we say that a distribution $\mathcal{D}$ of configurations \emph{represents} the SDP solution if it is obtained as follows:
    \begin{itemize}
        \item Sample a constraint $C = (P, i, j) \sim \mathcal{C}$ from the instance.
        \item Return the configuration $\theta = (b_i, b_j, b_{ij}, P)$.
    \end{itemize}
    We may also say that $\theta$ as defined above represents the constraint $C$.
\end{definition}

Now we define \emph{completeness}, which corresponds to the value of the Basic SDP relaxation.
\begin{definition}[Completeness]
    Given a configuration $\theta = (b_i, b_j, b_{ij}, P)$, we define its \emph{completeness} to be 
    \begin{equation}
        \comp(\theta) = \hat{P}_\varnothing + \hat{P}_{i} b_i + \hat{P}_j b_j + \hat{P}_{i,j} b_{ij}.
    \end{equation}
    For a distribution $\mathcal{D}$ over configurations, we define $\comp(\mathcal{D}) = \E_{\theta \sim \mathcal{D}}[\comp(\theta)]$.
\end{definition}

By construction, if a distribution $\mathcal{D}$ represents some SDP solution, then its completeness corresponds exactly to the SDP objective value achieved by that SDP solution.

\subsection{\texorpdfstring{The $\THRESH^-$ Rounding Family}{The THRESH- Rounding Family}}

The $\THRESH^-$ rounding family, introduced in~\cite{lewin2002improved}, is a powerful family of rounding algorithms for Boolean MAX 2-CSPs. A rounding algorithm from the $\THRESH^-$ rounding family is specified by a \emph{threshold function} $f: [-1, 1] \to \mathbb{R}$. Given an SDP solution $\bv_0, \bv_1, \ldots, \bv_n$, it computes the projection $\bv_i^\perp = \frac{\bv_i - b_i \bv_0}{\sqrt{1 - b_i^2}}$ for every $i \in [n]$, samples a standard Gaussian vector $\br$, and sets the variable $i$ to be $\true$ if $\br \cdot \bv_i^\perp \geq f(b_i)$ and $\false$ otherwise.

We now define \emph{soundness}, which corresponds to the probability that a configuration (or more precisely, any constraint represented by the configuration) is satisfied using a $\THRESH^-$ rounding algorithm.

\begin{definition}[Soundness]
    Let $f: [-1, 1] \to \mathbb{R}$ be a threshold function. Given a configuration $\theta = (b_i, b_j, b_{ij}, P)$, let $(g_i, g_j)$ be a pair of standard Gaussian random variables with
    correlation $\rho(\theta)$, and define
    \[
        X_i = 
        \begin{cases}
            -1, & \text{if } g_i \geq f(b_i),\\
            1, & \text{otherwise,}
        \end{cases}
        \qquad
        X_j = 
        \begin{cases}
            -1, & \text{if } g_j \geq f(b_j),\\
            1, & \text{otherwise,}
        \end{cases}
    \]
    Then, the \emph{soundness} of $\theta$ with respect to $f$ is defined to be
    \[
        \sound_f(\theta)
        =
        \Pr[P(X_i,X_j)=\true].
    \]
    For a distribution $\mathcal{D}$ over configurations, we define $\sound_f(\mathcal{D}) = \E_{\theta \sim \mathcal{D}}[\sound_f(\theta)]$.
\end{definition}

By construction, if a distribution $\mathcal{D}$ represents some SDP solution, then its soundness with respect to $f$ corresponds to the expected satisfied fraction if we use $\THRESH^-$ with $f$ to round that SDP solution.

\section{Description of the Algorithm}\label{sec:algo}

We now state our main algorithm. It is adapted from the bounded-degree MAX CUT algorithm due to~\cite{hsieh2023approximating}. The quantities $W_{i, 1}, W_{i, 0}$, and $W_{i, \unassigned}$ will be defined momentarily after the algorithm.

\begin{algorithm}[H]
    \caption{Algorithm for Bounded-Degree Boolean MAX 2-CSP}\label{alg:main}
    \begin{algorithmic}[1]
    \State \textbf{Input:} Boolean MAX 2-CSP instance $I = (V, \mathcal{C}, w)$ with variable set $V = [n]$, parameter $\epsilon > 0$, and an SDP solution $\{\bv_i \mid i \in V \cup \{0\}\}$.
    \State \textbf{Output:} an assignment $A: V \to \{\true, \false\}$.
    \medskip
    \State Choose a standard Gaussian vector $\br$.
    \For{$i \in V$}
        \State Compute $b_i = \bv_0 \cdot \bv_i$ and $\bv_i^\perp = \frac{\bv_i - b_i \bv_0}{\sqrt{1 - b_i^2}}$.
        \State Set $A'(i) = \true$ if $\br \cdot \bv_i^\perp \geq f(b_i)$ and $\false$ otherwise.
        \State Set $A(i) = A'(i)$ if $\left|\br \cdot \bv_i^\perp - f(b_i)\right| \geq \epsilon$ and $\unassigned$ otherwise.
    \EndFor
    \For{$i \in V$} \label{Line:flip_for_loop}
        \If{$A(i) = \unassigned$}
            \State Set $A(i) = \left\{\begin{array}{ll}
                \true & \text{if } W_{i, 1} > W_{i, 0} + W_{i, \unassigned}, \\
                \false & \text{if } W_{i, 0} > W_{i, 1} + W_{i, \unassigned}, \\
                A'(i) & \text{otherwise.}
            \end{array}\right.$\label{Line:flipping}
        \EndIf
    \EndFor
    \State \Return $A$
    \end{algorithmic}
\end{algorithm}

Informally, the algorithm first computes an assignment $A'$ using $\THRESH^-$ with the threshold function $f$. It then greedily flips any variable which lies approximately on the boundary of the $\THRESH^-$ rounding algorithm (more specifically, $\left|\br \cdot \bv_i^\perp - f(b_i)\right| < \epsilon$) if there is some \emph{local gain}, which we define below.

For every variable $i$, let $\mathcal{C}_i \subseteq \mathcal{C}$ be the set of constraints that involve $i$. We define the following subsets of $\mathcal{C}_i$ with respect to the (partial) assignment $A$ before executing the for-loop at Line~\ref{Line:flip_for_loop}:
\begin{align*}
\mathcal{C}_{i, \unassigned} & = \{C \in \mathcal{C}_i \mid \text{the other variable in }C \text{ is unassigned}\} \\
\mathcal{C}_{i, \sat} & = \{C \in \mathcal{C}_i \setminus \mathcal{C}_{i, \unassigned} \mid C \text{ is satisfied regardless of the truth value of } i\} \\
\mathcal{C}_{i, \unsat} & = \{C \in \mathcal{C}_i\setminus \mathcal{C}_{i, \unassigned}  \mid C \text{ is unsatisfied regardless of the truth value of } i\} \\
\mathcal{C}_{i, 1} & = \{C \in \mathcal{C}_i \setminus \mathcal{C}_{i, \unassigned} \mid C \text{ is satisfied if } i \text{ is } \true \text{ and unsatisfied if } i \text{ is } \false\} \\
\mathcal{C}_{i, 0} & = \{C \in \mathcal{C}_i \setminus \mathcal{C}_{i, \unassigned} \mid C \text{ is satisfied if } i \text{ is } \false \text{ and unsatisfied if } i \text{ is } \true\}.
\end{align*}
We use $W_{i, \unassigned}, W_{i, \sat}, W_{i, \unsat}, W_{i, 1}, W_{i, 0}$ to denote the total weights of clauses in these subsets, and we define
\[
W_i \coloneqq \sum_{C \in \mathcal{C}_i}w(C) = W_{i, \unassigned} + W_{i, \sat} + W_{i, \unsat} + W_{i, 1} + W_{i, 0}
\]
to be the total weight of constraints which contain $i$.
We would like to understand how the for-loop starting at Line~\ref{Line:flip_for_loop} affects the constraints in these subsets. By definition, constraints in $\mathcal{C}_{i, \sat}$ and $\mathcal{C}_{i, \unsat}$ won't be affected by the truth value for $i$ since they are already fixed, and therefore they (as well as the weights $W_{i, \sat}, W_{i, \unsat}$) won't play a role in our analysis. The only constraints whose status can be affected in the for-loop are the ones in $\mathcal{C}_{i, \unassigned}, \mathcal{C}_{i, 0}$ and $\mathcal{C}_{i, 1}$. If we set $i$ to be $\true$, then we satisfy every constraint in $\mathcal{C}_{i, 1}$, possibly some constraints in $\mathcal{C}_{i, \unassigned}$, and certainly none in $\mathcal{C}_{i, 0}$. Similarly, if we set $i$ to be $\false$, then we satisfy every constraint in $\mathcal{C}_{i, 0}$, possibly some constraints in $\mathcal{C}_{i, \unassigned}$, and certainly none in $\mathcal{C}_{i, 1}$. It follows that if we flip $i$ from $\true$ to $\false$, then the total weight of satisfied constraints changes by $W_{i, 0} - W_{i, 1} - p\cdot W_{i, \unassigned}$, for some $p \in [-1, 1]$. If we have $W_{i, 0} - W_{i, 1} - W_{i, \unassigned} > 0$, then this change is guaranteed to be positive. In particular, this means when $W_{i, 0} - W_{i, 1} - W_{i, \unassigned} > 0$ and $A'(i) = \true$, we obtain an improvement over $A'$ by setting $A(i) = \false$. Analogously, when $W_{i, 1} - W_{i, 0} - W_{i, \unassigned} > 0$ and $A'(i) = \false$, we also obtain an improvement over $A'$ by setting $A(i) = \true$.
This motivates the formal definition of \emph{local gains} as follows: 
\begin{definition}
    We define the local gain $\Delta_i$ of the variable $i$ to be 
\begin{align*}
\Delta_i & = (W_{i, 0} - W_{i, 1} - W_{i, \unassigned})_+  \cdot \bone\{\br \cdot \bv_i^\perp \in (f(b_i), f(b_i) + \epsilon)\} \\
& \qquad + (W_{i, 1} - W_{i, 0} - W_{i, \unassigned})_+  \cdot \bone\{\br \cdot \bv_i^\perp \in (f(b_i) - \epsilon, f(b_i))\}.
\end{align*}
Here $(z)_+$ is defined to be $z$ if $z \geq 0$ and $0$ if $z < 0$.
\end{definition}

Note that there is tension over the choice of $\epsilon$, which controls how many vertices we potentially flip in the for-loop: on one hand we would like to flip more vertices in order to get more local gain, on the other hand we would like $W_{i, \unassigned}$ to be small, for otherwise it may overwhelm the difference between $W_{i, 0}$ and $W_{i, 1}$ and we won't get any local gain at all. 

The following lemma formalizes the motivation for local gains.

\begin{lemma}\label{lem:val_diff_local_gain}
    Let $\mathcal{D}$ be the distribution of configurations that represents the SDP solution in the input to Algorithm~\ref{alg:main}. Then, we have $\E[\Val(I, A)] \geq \sound_f(\mathcal{D}) + \sum_{i \in V} \E[\Delta_i]$.
\end{lemma}
\begin{proof}
    Observe that for any $i \in V$, we have $A'(i) = A(i)$ unless 
    \begin{itemize}
        \item $W_{i, 1} > W_{i, 0} + W_{i, \unassigned}$ and $\br \cdot \bv_i^\perp \in (f(b_i) - \epsilon, f(b_i))$, or
        \item $W_{i, 0} > W_{i, 1} + W_{i, \unassigned}$ and $\br \cdot \bv_i^\perp \in (f(b_i), f(b_i) + \epsilon)$.
    \end{itemize}
    In these two cases, we flip the value of variable $i$ and the weight of satisfied constraints increases by at least $\Delta_i$. So we have
    \[\Val(I, A) - \Val(I, A') \geq \sum_{i \in V} \Delta_i.\]
    The lemma follows by taking expectation and observing that $\E[\Val(I, A')] = \sound_f(\mathcal{D})$.
\end{proof}

\section{Analysis for MAX 2-SAT}\label{sec:2-sat}

We first analyze Algorithm~\ref{alg:main} in the case of MAX 2-SAT. Let $\mathcal{P}_{2SAT} = \{x \vee y, \overline{x} \vee y, \overline{x} \vee \overline{y}\}$, then MAX 2-SAT is exactly $\MAXCSP(\mathcal{P}_{2SAT})$. We start by stating some known results for MAX 2-SAT. For $\beta \in [0, 1]$, define the threshold function $f_\beta: [-1, 1] \to \mathbb{R}$ by
\begin{equation}
    f_\beta(t) = \Phi^{-1}\left(\frac{1 + \beta t}{2}\right),
\end{equation}
where $\Phi^{-1}$ is the inverse CDF for a standard Gaussian variable. Observe that $f_\beta$ is odd for any $\beta \in [0, 1]$, in the sense that $f_\beta(-t) = -f_\beta(t)$ for every $t \in [-1, 1]$. 

\begin{theorem}[\cite{lewin2002improved,austrin2007balanced,brakensiek2024tight}]
    There exists some constant $\beta_\star \approx 0.94016567$ such that for any MAX~2-SAT configuration $\theta$,
    \begin{equation}
        \sound_{f_{\beta_\star}}(\theta) - \beta_\ast \cdot \comp(\theta) \geq 0.
    \end{equation}Consequently, the $\THRESH^-$ algorithm with the $f = f_{\beta_\star}$ achieves a $\beta_\star$-approximation for MAX~2-SAT. On the other hand, assuming the Unique Games Conjecture, for any $\epsilon > 0$ it is NP-hard to approximate MAX 2-SAT within a factor of $\beta_\star + \epsilon$.
\end{theorem}

By inspecting the proofs in~\cite{brakensiek2024tight} we in fact have the following slightly stronger statement:
\begin{theorem}[\cite{brakensiek2024tight}]\label{thm:2sat_stronger}
    Given a valid triple $(b_i, b_j, b_{ij})$, define 
    \[
    g(b_i, b_j, b_{ij}) = \sound_{f_{\beta_\star}}(\theta) - \beta_\ast \cdot \comp(\theta),
    \]
    where $\theta = (b_i, b_j, b_{ij}, x \vee y)$.
    Then $g$ has exactly two minimizers 
    \[g(b_\star, b_\star, -1 + 2b_\star) = g(-b_\star, -b_\star, -1 + 2b_\star) = 0\] where $b_\star \approx 0.162478$. Furthermore, there exists $c > 0$ such that for any valid tuple $(b_i, b_j, b_{ij})$,
    \[
    g(b_i, b_j, b_{ij}) \geq c \cdot \left(\min\left(\|(b_i, b_j, b_{ij}) - (b_\star, b_\star, -1 + 2b_\star)\|, \|(b_i, b_j, b_{ij}) - (-b_\star, -b_\star, -1 + 2b_\star)\|\right)\right)^2.
    \] 
\end{theorem}

Let $\theta = (b_i, b_j, b_{ij}, P)$ be a MAX 2-SAT configuration (i.e., $P \in \mathcal{P}_{2SAT}$) and $\delta > 0$. We say that $\theta$ is $\delta$-critical, if one of the following holds:
\begin{itemize}
    \item $P = x \vee y$ or $P = \overline{x} \vee \overline{y}$, and
    \begin{equation}
    \min\left(\|(b_i, b_j, b_{ij}) - (b_\star, b_\star, -1 + 2b_\star)\|, \|(b_i, b_j, b_{ij}) - (-b_\star, -b_\star, -1 + 2b_\star)\|\right) \leq \delta.
    \end{equation}
    \item $P = \overline{x} \vee y$ and
    \begin{equation}
    \min\left(\|(b_i, b_j, b_{ij}) - (b_\star, -b_\star, 1 - 2b_\star)\|, \|(b_i, b_j, b_{ij}) - (-b_\star, b_\star, 1 - 2b_\star)\|\right) \leq \delta.
    \end{equation}
\end{itemize}

Theorem~\ref{thm:2sat_stronger} shows that for non-critical configurations we may hope for a better-than-$\beta_\star$ approximation.
\begin{lemma}\label{lem:not_critical}
    Let $\theta$ be a MAX 2-SAT configuration $\theta$ that is not $\delta$-critical. Then
    \begin{equation}
        \sound_{f_{\beta_\star}}(\theta) - \beta_\star \cdot \comp(\theta) = \Omega(\delta^2).
    \end{equation}
\end{lemma}
\begin{proof}
    Follows directly from Theorem~\ref{thm:2sat_stronger}.
\end{proof}

On the other hand, for $\delta$-critical configurations with sufficiently small $\delta > 0$, we may use Algorithm~\ref{alg:main} to obtain an improvement if the instance is bounded degree.

\begin{theorem}\label{thm:2-SAT-critical}
    There exists some constant $\delta > 0$ such that the following holds. Let $I = (V, \mathcal{C}, w)$ be an instance of MAX 2-SAT such that every variable appears in at most $d$ constraints. Let $\mathcal{D}$ be a distribution of configurations representing some SDP solution $\{\bv_i \mid i \in V \cup \{0\}\}$ for $I$ such that $b_i \geq 0$ for every variable $i$, and every $\theta \in \supp(\mathcal{D})$ is $\delta$-critical. Then, for $f = f_{\beta_\star}$ and $\epsilon = d^{-1}\log^{-1}d$, Algorithm~\ref{alg:main} returns an assignment $A$ such that $\E[\Val(I, A)] = \sound_{f}(\mathcal{D}) + \Omega(d^{-2}\log^{-1}(d))$.
\end{theorem}

We remark that the assumption that $b_i \geq 0$ for every variable $i$ is for convenience only. The proof of Theorem~\ref{thm:2-SAT-critical} is deferred to Section~\ref{subsec:local_2sat}. Combining Lemma~\ref{lem:not_critical} and Theorem~\ref{thm:2-SAT-critical} we obtain the following theorem:

\begin{theorem}[Restatement of Theorem \ref{thm:intro-2}]\label{thm:2sat_bded_degree}
    Let $I$ be an instance of MAX 2-SAT such that every variable appears in at most $d$ constraints. Let $\mathcal{D}$ be a distribution of configurations representing an SDP solution for $I$. Then, for $f = f_{\beta_\star}$, we can find an assignment $A$ such that \[\Val(I, A) \geq \beta_\star \cdot \comp(\mathcal{D}) + \Omega(d^{-2}\log^{-1}(d)).\]
\end{theorem}
\begin{proof}
    Let $\delta > 0$ be chosen as in Theorem~\ref{thm:2-SAT-critical} and $\eta = cd^{-2}\log^{-1}(d)$ for some small constant $c > 0$. If there are at least $\eta$-fraction (in the weighted sense) of configurations in $\mathcal{D}$ that are not $\delta$-critical, then in this case by Lemma~\ref{lem:not_critical} we have that the $\THRESH^-$ algorithm with the $f = f_{\beta_\star}$ already achieves $\sound_{f}(\mathcal{D}) \geq \beta_\star\cdot \comp(\mathcal{D}) +\Omega(\eta) = \beta_\star\cdot \comp(\mathcal{D}) +\Omega(d^{-2}\log^{-1}(d))$.

    Now assume at least $(1 - \eta)$-fraction of the configurations are $\delta$-critical, then we may ignore the constraints corresponding to configurations that are not $\delta$-critical and form a new instance $I'$ and a new distribution of configurations $\mathcal{D}'$ supported only on $\delta$-critical configurations by renormalizing the weights.
    Applying Theorem~\ref{thm:2-SAT-critical} to $I'$  and $\mathcal{D}'$ (after possibly flipping variables and SDP vectors so that $b_i \geq 0$ for every $i$), we obtain an assignment achieving a satisfying fraction of at least
    \begin{align*}
        & \left(1 - \eta\right)\cdot\left( \sound_{f}(\mathcal{D}') + \Omega(d^{-2}\log^{-1}(d))\right) \\
        \geq\, & \left(1 - \eta\right)\cdot\left( \beta_\star \cdot \comp(\mathcal{D}') + \Omega(d^{-2}\log^{-1}(d))\right)\\
        \geq\, & \left(1 - \eta\right)\cdot\left( \beta_\star \cdot (\comp(\mathcal{D}) - \eta) + \Omega(d^{-2}\log^{-1}(d))\right) \\
        =\, & \beta_\star \cdot \comp(\mathcal{D}) + \Omega(d^{-2}\log^{-1}(d)) - \eta \cdot O(1)\\
        =\, & \beta_\star \cdot \comp(\mathcal{D}) + \Omega(d^{-2}\log^{-1}(d)).\qedhere
    \end{align*}
\end{proof}

\subsection{Analyzing a Local Neighborhood}\label{subsec:local_2sat}
In this subsection we aim to prove Theorem~\ref{thm:2-SAT-critical}. Since we assume $b_i \geq 0$ for every variable $i$, our $\delta$-critical configurations $(b_i, b_j, b_{ij}, P)$ must all satisfy that $\|(b_i, b_j, b_{ij}) - (b_\star, b_\star, -1 + 2b_\star)\| \leq \delta$ and $P$ is either $x \vee y$ or $\overline{x} \vee \overline{y}$, if we take $\delta$ to be sufficiently small. 

To prove Theorem~\ref{thm:2-SAT-critical}, we need a few lemmas on the algorithmic quantities $W_{i, 0}, W_{i, 1}, W_{i, \unassigned}, \Delta_i$ in the neighborhood of $i$ defined in Section~\ref{sec:algo}. To avoid ambiguity, we use $A_\unassigned$ to denote the (partial) assignment $A$ before the execution of the for-loop at Line~\ref{Line:flip_for_loop}. The overall idea is to show that conditioned on $i$ being unassigned, $W_{i, 0} - W_{i, 1}$ must have a large variance, leading to a nontrivial local gain.

\begin{lemma}\label{lem:2sat_1norm}
For every $i$ and $t \in (f(b_i) - \epsilon, f(b_i) + \epsilon)$, we have \[\E[|W_{i, 0} - W_{i, 1}| \mid \br \cdot \bv_i^\perp = t] = W_i \cdot  \Omega(d^{-1}),\]
where the constant in $\Omega$ can be chosen independently from $t$.
\end{lemma}
\begin{proof}
Let $\mathcal{C}_i^+ \subseteq \mathcal{C}_i$ be the set of constraints involving $i$ that are defined using the predicate $x \vee y$, and $\mathcal{C}_i^- \subseteq \mathcal{C}_i$ be the set of those defined using $\overline{x} \vee \overline{y}$.
 For $C \in \mathcal{C}_i$, let $j_C$ be the variable other than $i$ involved in $C$. For every $C \in \mathcal{C}_i^+$, we have 
\[
C\in \left\{\begin{array}{ll}
\mathcal{C}_{i, \sat} & \text{if } A_\unassigned(j_C) = \true,\\
\mathcal{C}_{i, 1} & \text{if } A_\unassigned(j_C) = \false,\\
\mathcal{C}_{i, \unassigned} & \text{if } A_\unassigned(j_C) = \unassigned.
\end{array}\right.
\]
Similarly, for every $C \in \mathcal{C}_i^-$, we have
\[
C\in \left\{\begin{array}{ll}
\mathcal{C}_{i, \sat} & \text{if } A_\unassigned(j_C) = \false,\\
\mathcal{C}_{i, 0} & \text{if } A_\unassigned(j_C) = \true,\\
\mathcal{C}_{i, \unassigned} & \text{if } A_\unassigned(j_C) = \unassigned.
\end{array}\right.
\]
So we have 
\begin{equation}\label{eq:w1-w0}
    W_{i, 0} - W_{i, 1} = \sum_{C \in \mathcal{C}_i^+} w(C)\bone\{A_\unassigned(j_C) = \true\} - \sum_{C \in \mathcal{C}_i^-} w(C)\bone\{A_\unassigned(j_C) = \false\}.
\end{equation}

Given a configuration $\theta = (b_i, b_j, b_{ij}, P)$, we have $A_\unassigned(j) = \true$ if and only if $\br \cdot \bv_j^\perp \geq f(b_j) + \epsilon$, or equivalently $\br \cdot \bv_{j}^{\perp \{0, i\}} \geq \frac{-\rho(\theta)\br \cdot \bv_i^\perp + f(b_j) + \epsilon}{\sqrt{1 - \rho(\theta)^2}}$. Similarly $A_\unassigned(j) = \false$ if and only if $\br \cdot \bv_{j}^{\perp \{0, i\}} \leq \frac{-\rho(\theta)\br \cdot \bv_i^\perp + f(b_j) - \epsilon}{\sqrt{1 - \rho(\theta)^2}}$. Since $\br \cdot \bv_i^\perp \in (f(b_i) - \epsilon, f(b_i) + \epsilon)$ and $\theta$ is $\delta$-critical, we have
$\|\frac{-\rho(\theta)\br \cdot \bv_i^\perp + f(b_j) \pm \epsilon}{\sqrt{1 - \rho(\theta)^2}} - s\| = O(\epsilon + \delta)$, where
\begin{align*}
    s \coloneqq \frac{-\rho_\star f(b_\star) + f(b_\star)}{\sqrt{1 - \rho_\star^2}} = f(b_\star) \cdot \frac{1 - \rho_\star}{\sqrt{1 - \rho_\star^2}} = f(b_\star) \cdot \sqrt{\frac{1 - \rho_\star}{1 + \rho_\star}}
\end{align*}
and $\rho_\star = \frac{-1 + 2b_\star - b_\star^2}{1 - b_\star^2} = -\frac{1-b_\star}{1 + b_\star}$. By plugging in the value of $b_\star$ and $\beta_\star$, we get $s \approx 0.4779$. So we may rewrite
\begin{align*}
    W_{i, 0} - W_{i, 1} & = \sum_{C \in \mathcal{C}_i^+} w(C)\bone\{\br \cdot \bvperp{j_C}{0, i} \geq s \pm O(\epsilon + \delta)\} - \sum_{C \in \mathcal{C}_i^-} w(C)\bone\{\br \cdot \bvperp{j_C}{0, i} \leq s \pm O(\epsilon + \delta)\} \\
    & = \sum_{C \in \mathcal{C}_i} w(C)\bone\{\br \cdot \bvperp{j_C}{0, i} \geq s \pm O(\epsilon + \delta)\} - \sum_{C \in \mathcal{C}_i^- } w(C).
\end{align*}
In the last step we used the fact that for a continuous random variable $X$ we have $\bone\{X \leq s\} = 1 - \bone\{X \geq s\}$. In the above calculations, the terms $O(\epsilon + \delta)$ may be different depending on $C$ but we may choose them small enough so that they are all less than $s / 2$. Note that $\sum_{C \in \mathcal{C}_i^- } w(C)$ is a constant which does not affect the variance. So, we can apply Lemma~\ref{lem:variance_2sat} and obtain 
\[\Var[W_{i, 0} - W_{i, 1} \mid \br \cdot \bv_i^\perp = t] = W_i^2 \cdot \Omega(d^{-1}).\] It follows that 
\begin{align*}
    \E[|W_{i, 0} - W_{i, 1}| \mid \br \cdot \bv_i^\perp = t] & \geq \frac{1}{W_i}\cdot \E[(W_{i, 0} - W_{i, 1})^2 \mid \br \cdot \bv_i^\perp = t]  \\ & \geq \frac{1}{W_i}\cdot \Var[W_{i, 0} - W_{i, 1} \mid \br \cdot \bv_i^\perp = t] \\
    & = W_i \cdot  \Omega(d^{-1}).\qedhere
\end{align*}
\end{proof}

\begin{lemma}\label{lem:2-sat_unassigned}
    For every $i$ and $t \in (f(b_i) - \epsilon, f(b_i) + \epsilon)$ we have $\E\left[W_{i, \unassigned} \mid \br \cdot \bv_i^\perp = t\right] = \Theta(\epsilon \cdot W_i)$ where the constant in $\Theta$ can be chosen independently from $t$.
\end{lemma}
\begin{proof}
    Consider a constraint $C \in \mathcal{C}_i$ involving $i$ and some other variable $j$, represented by the configuration $\theta = (b_i, b_j, b_{ij}, P)$. As computed in the proof of Lemma~\ref{lem:2sat_1norm}, we have that $j$ is unassigned if and only if
    \begin{equation}\label{eq:threshold_for_bv_j_perp}
     \frac{-\rho(\theta)\br \cdot \bv_i^\perp + f(b_j) - \epsilon}{\sqrt{1 - \rho(\theta)^2}} < \br \cdot \bvperp{j}{0, i} < \frac{-\rho(\theta)\br \cdot \bv_i^\perp + f(b_j) + \epsilon}{\sqrt{1 - \rho(\theta)^2}}.
    \end{equation}
    This is an interval of length $\Theta(\epsilon)$. It follows that 
    \[\E\left[\bone\{C \in \mathcal{C}_{i, \unassigned}\} \mid \br\cdot\bv_i^\perp = t\right] = \Theta(\epsilon),
    \] and therefore
    \begin{equation*}
    \E\left[W_{i, \unassigned} \mid \br \cdot \bv_i^\perp = t\right] = \E\left[\sum_{C \in \mathcal{C}}w(C)\bone\{C \in \mathcal{C}_{i, \unassigned}\} \mid \br \cdot \bv_i^\perp = t\right] =\Theta(\epsilon\cdot W_i). \qedhere
    \end{equation*}
\end{proof}

\begin{lemma}[Local gains]\label{lem:2sat_local_gain}
    Choose $\epsilon = cd^{-1}\log^{-1}(d)$ for some sufficiently small constant $c > 0$. Then, for every variable $i$, we have $\E[\Delta_i \mid A_\unassigned(i) = \unassigned] = W_i \cdot \Omega(d^{-1})$. 
\end{lemma}
\begin{proof}
    It is sufficient to show that  
    \[\E[\Delta_i \mid \br \cdot \bv_i^\perp = t] = W_i \cdot \Omega(d^{-1})\]
    either for every $t \in (f(b_i) - \epsilon, f(b_i))$ or for every $t \in (f(b_i), f(b_i) + \epsilon)$.

    From the definition
    \begin{align*}\Delta_i & = (W_{i, 0} - W_{i, 1} - W_{i, \unassigned})_+  \cdot \bone\{\br \cdot \bv_i^\perp \in (f(b_i), f(b_i) + \epsilon)\} \\
& \qquad + (W_{i, 1} - W_{i, 0} - W_{i, \unassigned})_+  \cdot \bone\{\br \cdot \bv_i^\perp \in (f(b_i) - \epsilon, f(b_i))\}\end{align*}
    we have 
    \begin{align*}
        & \E[\Delta_i \mid \br \cdot \bv_i^\perp \in (f(b_i), f(b_i) + \epsilon)] \\
        =\, &\E\left[(W_{i, 0} - W_{i, 1} - W_{i, \unassigned})_+ \mid \br \cdot \bv_i^\perp \in (f(b_i), f(b_i) + \epsilon)\right] \\
        \geq \, &\E\left[(W_{i, 0} - W_{i, 1})_+ - W_{i, \unassigned} \mid \br \cdot \bv_i^\perp \in (f(b_i), f(b_i) + \epsilon)\right] \\
        = \, &\E\left[\frac{|W_{i, 0} - W_{i, 1}| + (W_{i, 0} - W_{i, 1})}{2}- W_{i, \unassigned} \mid \br \cdot \bv_i^\perp \in (f(b_i), f(b_i) + \epsilon)\right],
    \end{align*}
    and similarly
    \begin{align*}
        & \E[\Delta_i \mid \br \cdot \bv_i^\perp \in (f(b_i) - \epsilon, f(b_i))] \\
        \geq \, &\E\left[\frac{|W_{i, 1} - W_{i, 0}| + (W_{i, 1} - W_{i, 0})}{2}- W_{i, \unassigned} \mid \br \cdot \bv_i^\perp \in (f(b_i) - \epsilon, f(b_i))\right].
    \end{align*}

    Assume that $\E[W_{i, 0} - W_{i, 1} \mid \br \cdot \bv_i^\perp \in (f(b_i), f(b_i) + \epsilon)] \geq -W_i / (d \sqrt{\log d})$. Then by Lemma~\ref{lem:2sat_1norm} and Lemma~\ref{lem:2-sat_unassigned} we have 
    \begin{align*}
        & \E\left[\frac{|W_{i, 0} - W_{i, 1}| + (W_{i, 0} - W_{i, 1})}{2}- W_{i, \unassigned} \mid \br \cdot \bv_i^\perp \in (f(b_i), f(b_i) + \epsilon)\right] \\
        \geq\, & W_i \cdot \left(\Omega(d^{-1}) - d^{-1}\log^{-1/2}(d) - O(d^{-1}\log^{-1}(d))\right) \\
        =\, & W_i \cdot \Omega(d^{-1}).
    \end{align*}
    Otherwise, we have $\E[W_{i, 0} - W_{i, 1} \mid \br \cdot \bv_i^\perp \in (f(b_i), f(b_i) + \epsilon)] < -W_i / (d \sqrt{\log d})$, or equivalently $\E[W_{i, 1} - W_{i, 0} \mid \br \cdot \bv_i^\perp \in (f(b_i), f(b_i) + \epsilon)] > W_i / (d \sqrt{\log d})$. Observe that by Equation~\eqref{eq:w1-w0} the function $\E[W_{i, 1} - W_{i, 0} \mid \br \cdot \bv_i^\perp = t]$ is $O(W_i)$-Lipschitz in $t$, so we have
    \[
    \E[W_{i, 1} - W_{i, 0} \mid \br \cdot \bv_i^\perp \in (f(b_i) -\epsilon, f(b_i))] > W_i / (d \sqrt{\log d}) - \epsilon \cdot O(W_i) = W_i \cdot \Omega(d^{-1}\log^{-1/2}(d)).
    \]
    It follows that 
    \begin{align*}
    & \E\left[\frac{|W_{i, 1} - W_{i, 0}| + (W_{i, 1} - W_{i, 0})}{2}- W_{i, \unassigned} \mid \br \cdot \bv_i^\perp \in (f(b_i) - \epsilon, f(b_i))\right] \\
        \geq\, & W_i \cdot \left(\Omega(d^{-1}) + \Omega(d^{-1}\log^{-1/2}(d)) - O(d^{-1}\log^{-1}(d))\right) \\
        =\, & W_i \cdot \Omega(d^{-1}).
    \end{align*}
    Thus, we conclude that $\E[\Delta_i \mid A_\unassigned(i) = \unassigned] = W_i \cdot \Omega(d^{-1})$.
\end{proof}

We are now ready to prove Theorem~\ref{thm:2-SAT-critical}.

\begin{proof}[Proof of Theorem~\ref{thm:2-SAT-critical}]
    By Lemma~\ref{lem:val_diff_local_gain}, we have
    \begin{align*}
        \E[\Val(I, A)] & \geq \sound_f(\mathcal{D}) + \sum_{i \in V}\E[\Delta_i] \\
        & \geq \sound_f(\mathcal{D}) + \sum_{i \in V}\E[\Delta_i \mid A_\unassigned(i) = \unassigned] \cdot \Pr[A_\unassigned(i) = \unassigned] \\
        & \geq \sound_f(\mathcal{D}) + \sum_{i \in V}W_i \cdot \Omega(d^{-1}) \cdot \Omega(\epsilon) \\
        & \geq \sound_f(\mathcal{D}) + \Omega(d^{-2}\log^{-1}d).
    \end{align*}
    Here we used the assumption that $\sum_{i \in V}W_i = 2$ (since total weight is 1 and each constraint is counted twice) and Lemma~\ref{lem:2sat_local_gain}.
\end{proof}

\section{Analysis for General Boolean MAX 2-CSPs}\label{sec:2csp}

In this section we analyze general Boolean MAX 2-CSPs. To this end we need the distribution of configurations to be \emph{reasonable} as introduced below.

\begin{definition}
    We say that two predicates $P_1, P_2: \{-1, 1\}^2 \to \{0, 1\}$ are \emph{aligned} if there exist $\sigma_1, \sigma_2 \in \{-1, 1\}$ and $s \in \{0, 1\}$ such that 
    \begin{equation}
        P_1(-1, \sigma_1) = P_2(-\sigma_2, \sigma_1\sigma_2) = s \quad\text{ and }\quad P_1(1, \sigma_1) = P_2(\sigma_2, \sigma_1\sigma_2) = 1-s.
    \end{equation}
    We say that they are \emph{anti-aligned} if there exist $\sigma_1, \sigma_2 \in \{-1, 1\}$ and $s \in \{0, 1\}$ such that 
    \begin{equation}
        P_1(-1, \sigma_1) = P_2(\sigma_2, \sigma_1\sigma_2) = s \quad\text{ and }\quad P_1(1, \sigma_1) = P_2(-\sigma_2, \sigma_1\sigma_2) = 1-s.
    \end{equation}
\end{definition}

Informally, if we look at two constraints with one variable in common, then they are aligned if there is a way to fix the other variables to induce the same unary preference on the shared variable, and they are anti-aligned if they induce opposite unary preferences.

\begin{definition}\label{def:reasonable_dist}
    We say that a distribution of configurations $\mathcal{D}$ representing some SDP solution is \emph{reasonable} if the following conditions hold:
    \begin{itemize}
        \item For every $\theta = (b_i, b_j, b_{ij}, P) \in \supp(\mathcal{D})$, we have $|b_i|, |b_j| \leq 0.4$ and $|b_{ij}| \in [0.5, 0.9]$.
        \item If $\theta_1 = (b_i, b_j, b_{ij}, P_1), \theta_2 = (b_i, b_k, b_{ik}, P_2) \in \supp(\mathcal{D})$ represent some constraints $(P_1, i, j)$ and $(P_2, i, k)$, then $\rho(\theta_1)\rho(\theta_2) \geq 0$ if $P_1$ and $P_2$ are aligned and $\rho(\theta_1)\rho(\theta_2) \leq 0$ if $P_1$ and $P_2$ are anti-aligned.
    \end{itemize}
\end{definition}

\begin{remark}
    The conditions for reasonableness may seem somewhat arbitrary. Indeed, the constants are chosen only for convenience: with these values we are able to prove bounds on the covariances as in Lemma~\ref{lem:cov_lower_bound}. That said, it may be verified that the hardest (optimal under UGC) distributions for MAX CUT and MAX 2-SAT are both reasonable. Furthermore, the known (numerically) hardest distributions (though not known to be optimal) for MAX 2-AND and MAX DI-CUT due to~\cite{brakensiek2023separating} are also reasonable. 
\end{remark}

The main benefit of reasonableness is the following bound on covariances in a local neighborhood. The proof of this lemma is deferred to the end of this section.

\begin{lemma}\label{lem:cov_lower_bound}
    Let $\mathcal{D}$ be a reasonable distribution of configurations representing some SDP solution $\{\bv_i \mid i \in \{0\} \cup [n]\}$ which satisfies all triangle inequalities. Let $\theta_1 = (b_i, b_j, b_{ij}, P_1), \theta_2 = (b_i, b_k, b_{ik}, P_2) \in \supp(\mathcal{D})$ be two configurations representing some constraints $(P_1, i, j)$ and $(P_2, i, k)$. If $P_1$ and $P_2$ are aligned, then $\bvperp{j}{0, i} \cdot \bvperp{k}{0, i} \geq -5/6$. If $P_1$ and $P_2$ are anti-aligned, then $\bvperp{j}{0, i} \cdot \bvperp{k}{0, i} \leq 5/6$. 
\end{lemma}

Under the reasonableness condition, we prove the following theorem.

\begin{theorem}\label{thm:general}
    Let $\mathcal{P}$ be a set of Boolean predicates of arity 2 such that every $P \in \mathcal{P}$ depends on both coordinates. Let $d \in \mathbb{N}$, and $B > 0$. Let $I$ be an instance of $\MAXCSP(\mathcal{P})$ such that every variable appears in at most $d$ constraints. Let $\mathcal{D}$ be a reasonable distribution of configurations corresponding to an SDP solution satisfying all triangle inequalities. Then, for some $\epsilon = {\Theta}(d^{-2}\log^{-2}(d))$ and any $f: [-1, 1] \to [-B, B]$, Algorithm~\ref{alg:main} returns an assignment to $I$ that satisfies at least $\sound_{f}(\mathcal{D}) + {\Omega}(d^{-4}\log^{-7/2}d)$ fraction of constraints, where the constant inside the ${\Omega}$ notation depends only on $\mathcal{P}$ and $B$.
\end{theorem} 
We will implicitly assume the assumptions in Theorem~\ref{thm:general} in our analysis for the remainder of the section. We use $A_\unassigned$ to denote the (partial) assignment $A$ before the execution of the for-loop at Line~\ref{Line:flip_for_loop}.  The overall idea is similar to Theorem~\ref{thm:2-SAT-critical}: we show that $W_{i, 0} - W_{i, 1}$ has a large variance which overwhelms the expected value of $W_{i, \unassigned}$, and therefore at least one of the two terms in $\Delta_i$ will have a nontrivial contribution. The variance lower bound for $W_{i, 0} - W_{i, 1}$ is shown in the following lemma.

\begin{lemma}\label{lem:W0-W1}
    There exist $c, \epsilon_0 > 0$ (depending only on the configurations and the threshold function $f$) such that for any $\epsilon < \epsilon_0$, $i \in V$ and  any $t \in (f(b_i) - \epsilon, f(b_i) + \epsilon)$, we have
    \[
    \E[|W_{i, 0} - W_{i, 1}| \mid \br \cdot \bv_i^\perp = t] \geq \frac{c\cdot W_i}{d^2\log^{3/2}d}.
    \]
\end{lemma}
\begin{proof}
    We have 
    \begin{equation}\label{eq:W_as_indicators}
    W_{i, 0} - W_{i, 1} = \sum_{C \in \mathcal{C}_i}w(C)\left(\bone\{C \in \mathcal{C}_{i, 0}\} - \bone\{C \in \mathcal{C}_{i, 1}\}\right).
    \end{equation}
    Let $C \in \mathcal{C}_i$ be a constraint on $i$ and some other variable $j$, represented by the configuration $\theta = (b_i, b_j, b_{ij}, P)$. 
    Observe that since by assumption $P$ depends on both variables, the indicator functions $\bone\{C \in \mathcal{C}_{i, 0}\}$ and  $\bone\{C \in \mathcal{C}_{i, 1}\}$ are two distinct functions among the three functions $\bone\{A_\unassigned(j) = \true\}$, $\bone\{A_\unassigned(j) = \false\}$, and 0 (as a constant function). Conditioned on $\br\cdot\bv_i^\perp = t$, we may write the first two indicators as $\bone\{\br \cdot \bvperp{j}{0, i}\geq s_j\}$, $\bone\{\br \cdot \bvperp{j}{0, i}\leq s_j'\}$, where $s_j$, $s_j'$ are some thresholds whose absolute values are bounded due to our reasonableness assumption (the expressions for these thresholds are the same as those in~\eqref{eq:threshold_for_bv_j_perp}). By using the fact that $\bone\{\br \cdot \bvperp{j}{0, i}\leq s_j'\} = 1 - \bone\{\br \cdot \bvperp{j}{0, i}\geq s_j'\}$, we may therefore rewrite \eqref{eq:W_as_indicators} as 
    \begin{equation}
    W_{i, 0} - W_{i, 1} = 
        \sum_{j \in J_1} w_j \bone\{\br \cdot \bvperp{j}{0, i} \geq s_j\} - 
        \sum_{j \in J_2} w_j \bone\{\br \cdot \bvperp{j}{0, i} \geq s_j'\} + c
    \end{equation}
    where $J_1, J_2$ are some (multi-)sets of variables and $c$ is some constant. We also have $\sum_{j \in J_1} w_j + \sum_{j \in J_2} w_j \in [W_i, 2W_i]$ since each constraint contributes either 1 or 2 nonzero indicator functions.

    Now we aim to bound the inner products between the $\bvperp{j}{0, i}$ vectors, which are equivalently the covariances between the standard Gaussian variables $\br \cdot \bvperp{j}{0, i}$. Take another constraint $C'$ involving $i$ and some third variable $k$, represented by $\theta' = (b_i, b_k, b_{ik}, P')$. Then $P$ and $P'$ are aligned if there exist $s \in \{0, 1\}$ and $\mathtt{s} \in \{\true, \false\}$ such that 
    \begin{equation}
        \bone\{C \in \mathcal{C}_{i, s}\} = \bone\{A_\unassigned(j) = \mathtt{s}\} \quad \text{ and } \quad
        \bone\{C' \in \mathcal{C}_{i, s}\} = \bone\{A_\unassigned(k) = \mathtt{s}\},
    \end{equation}
    or
    \begin{equation}
        \bone\{C \in \mathcal{C}_{i, s}\} = \bone\{A_\unassigned(j) = \mathtt{s}\} \quad \text{ and } \quad
        \bone\{C' \in \mathcal{C}_{i, 1-s}\} = \bone\{A_\unassigned(k) = \mathtt{not\,\,s}\},
    \end{equation}
    where we used $\mathtt{not\,\,s}$ to mean the truth value other than $\mathtt{s}$. These are exactly the cases where we get both $j, k \in J_1$ or $j, k \in J_2$.
    By Lemma~\ref{lem:cov_lower_bound}, we have $\bvperp{j}{0, i} \cdot \bvperp{k}{0, i} \geq -5/6$ if $P$ and $P'$ are aligned. On the other hand, $P$ and $P'$ are anti-aligned if there exist $s \in \{0, 1\}$ and $\mathtt{s} \in \{\true, \false\}$ such that 
    \begin{equation}
        \bone\{C \in \mathcal{C}_{i, s}\} = \bone\{A_\unassigned(j) = \mathtt{s}\} \quad \text{ and } \quad
        \bone\{C' \in \mathcal{C}_{i, s}\} = \bone\{A_\unassigned(k) = \mathtt{not\,\,s}\},
    \end{equation}
    or
    \begin{equation}
        \bone\{C \in \mathcal{C}_{i, s}\} = \bone\{A_\unassigned(j) = \mathtt{s}\} \quad \text{ and } \quad
        \bone\{C' \in \mathcal{C}_{i, 1-s}\} = \bone\{A_\unassigned(k) = \mathtt{s}\}.
    \end{equation}
    These are exactly the cases where one of $\{j, k\}$ is in $J_1$ and the other is in $J_2$.
    Again by Lemma~\ref{lem:cov_lower_bound}, we have $\bvperp{j}{0, i} \cdot \left(-\bvperp{k}{0, i}\right) \geq -5/6$ if $P$ and $P'$ are anti-aligned.
    
    Hence, by rewriting
    \begin{align*}
    W_{i, 0} - W_{i, 1} & = 
        \sum_{j \in J_1} w_j \bone\{\br \cdot \bvperp{j}{0, i} \geq s_j\} - 
        \sum_{j \in J_2} w_j \bone\{\br \cdot \bvperp{j}{0, i} \geq s_j'\} + c\\ 
        & = 
        \sum_{j \in J_1} w_j \bone\{\br \cdot \bvperp{j}{0, i} \geq s_j\} + 
        \sum_{j \in J_2} w_j \bone\{\br \cdot (-\bvperp{j}{0, i}) \geq -s_j'\} - \sum_{j \in J_2}w_j + c,
    \end{align*}
    we obtain an expression which is the weighted sum of indicators of the form $\bone\{g_j \geq s_j''\}$ (plus some constant) where $g_j$'s are correlated standard Gaussians and the covariances between them are all at least $-5/6$. We can therefore apply Lemma~\ref{lem:variance_main}, which gives us 
    \[
    \Var[W_{i, 0} - W_{i, 1} \mid \br \cdot \bv_i^\perp = t] \geq \Omega\left(\frac{ W_i^2}{d^2\log^{3/2}d}\right).
    \]
    Since $|W_{i, 0} - W_{i, 1}| \leq W_i$, we have
    \begin{equation*}
    \E[|W_{i, 0} - W_{i, 1}| \mid \br \cdot \bv_i^\perp = t] \geq \frac{1}{W_i}\E[(W_{i, 0} - W_{i, 1})^2 \mid \br \cdot \bv_i^\perp = t] \geq \Omega\left(\frac{W_i}{d^2\log^{3/2}d}\right). \qedhere
    \end{equation*}
\end{proof}

\begin{lemma}[Local gains]\label{lem:2csp_local_gain}
    Choose $\epsilon = cd^{-2}\log^{-2}(d)$ for some sufficiently small constant $c > 0$. Then, for every variable $i$, we have $\E[\Delta_i \mid A_\unassigned(i) = \unassigned] = W_i \cdot \Omega(d^{-2}\log^{-3/2}(d))$. 
\end{lemma}
\begin{proof}
As in the proof of Lemma~\ref{lem:2sat_local_gain}, it is sufficient to show that  
    \[\E[\Delta_i \mid \br \cdot \bv_i^\perp = t] = W_i \cdot \Omega(d^{-2}\log^{-3/2}(d))\]
    either for every $t \in (f(b_i) - \epsilon, f(b_i))$ or for every $t \in (f(b_i), f(b_i) + \epsilon)$.

As in Lemma~\ref{lem:2-sat_unassigned}, we have $\E\left[W_{i, \unassigned} \mid \br \cdot \bv_i^\perp = t\right] = \Theta(\epsilon \cdot W_i)$. Assume that $\E[W_{i, 0} - W_{i, 1} \mid \br \cdot \bv_i^\perp \in (f(b_i), f(b_i) + \epsilon)] \geq -W_i \cdot d^{-2}\cdot \log^{-2}(d)$. Then by Lemma~\ref{lem:W0-W1} we have 
    \begin{align*}
        & \E[\Delta_i \mid \br \cdot \bv_i^\perp \in (f(b_i), f(b_i) + \epsilon)] \\
        \geq\, & \E\left[\frac{|W_{i, 0} - W_{i, 1}| + (W_{i, 0} - W_{i, 1})}{2}- W_{i, \unassigned} \mid \br \cdot \bv_i^\perp \in (f(b_i), f(b_i) + \epsilon)\right] \\
        \geq\, & W_i \cdot \left(\Omega(d^{-2}\log^{-3/2}(d)) - d^{-2}\log^{-2}(d) - O(d^{-2}\log^{-2}(d))\right) \\
        =\, & W_i \cdot \Omega(d^{-2}\log^{-3/2}(d)).
    \end{align*}
    Otherwise, we have $\E[W_{i, 0} - W_{i, 1} \mid \br \cdot \bv_i^\perp \in (f(b_i), f(b_i) + \epsilon)] < -W_i \cdot d^{-2}\cdot \log^{-2}(d)$, or equivalently $\E[W_{i, 1} - W_{i, 0} \mid \br \cdot \bv_i^\perp \in (f(b_i), f(b_i) + \epsilon)] > W_i \cdot d^{-2}\cdot \log^{-2}(d)$. Since the function $\E[W_{i, 1} - W_{i, 0} \mid \br \cdot \bv_i^\perp = t]$ is $O(W_i)$-Lipschitz in $t$, we have
    \[
    \E[W_{i, 1} - W_{i, 0} \mid \br \cdot \bv_i^\perp \in (f(b_i) -\epsilon, f(b_i))] > W_i \cdot d^{-2} \log^{-2}(d) - \epsilon \cdot O(W_i) \geq -W_i \cdot O(d^{-2}\log^{-2}(d)).
    \]
    It follows that 
    \begin{align*}
        & \E[\Delta_i \mid \br \cdot \bv_i^\perp \in (f(b_i) - \epsilon, f(b_i))] \\
    \geq\, & \E\left[\frac{|W_{i, 1} - W_{i, 0}| + (W_{i, 1} - W_{i, 0})}{2}- W_{i, \unassigned} \mid \br \cdot \bv_i^\perp \in (f(b_i) - \epsilon, f(b_i))\right] \\
        \geq\, & W_i \cdot \left(\Omega(d^{-2}\log^{-3/2}(d)) - O(d^{-2}\log^{-2}(d)) - O(d^{-2}\log^{-2}(d))\right) \\
        =\, & W_i \cdot \Omega(d^{-2}\log^{-3/2}(d)).
    \end{align*}
    Thus, we conclude $\E[\Delta_i \mid A_\unassigned(i) = \unassigned] = W_i \cdot \Omega(d^{-2}\log^{-3/2}(d))$.
\end{proof}

\begin{proof}[Proof of Theorem~\ref{thm:general}]
    Follows similarly as in the proof of Theorem~\ref{thm:2-SAT-critical}. 
    By Lemma~\ref{lem:val_diff_local_gain}, we have
    \begin{align*}
        \E[\Val(I, A)] & \geq \sound_f(\mathcal{D}) + \sum_{i \in V}\E[\Delta_i] \\
        & \geq \sound_f(\mathcal{D}) + \sum_{i \in V}\E[\Delta_i \mid A_\unassigned(i) = \unassigned] \cdot \Pr[A_\unassigned(i) = \unassigned] \\
        & \geq \sound_f(\mathcal{D}) + \sum_{i \in V}W_i \cdot \Omega(d^{-2}\log^{-3/2}(d)) \cdot \Omega(\epsilon) \\
        & \geq \sound_f(\mathcal{D}) + \Omega(d^{-4}\log^{-7/2}d). \qedhere
    \end{align*}
\end{proof}

\subsection{Proof of Lemma~\ref{lem:cov_lower_bound}}

    We only prove the case where $P_1$ and $P_2$ are aligned since the anti-aligned case can be reduced to the aligned case by negating the variable $k$ as well as $\bv_k$. Assume that $P_1$ and $P_2$ are aligned. Then, since $\mathcal{D}$ is reasonable, we have $\rho(\theta_1)\rho(\theta_2) \geq 0$, $|b_i|, |b_j|, |b_k| \leq 0.4$ and $|b_{ij}|, |b_{ik}| \geq 0.5$. This implies that $|b_{ij}|\geq |b_i||b_j|$, and therefore $b_{ij}$ and $\rho(\theta_1) = \frac{b_{ij} - b_ib_j}{\sqrt{(1 - b_i^2)(1 - b_j^2)}}$ have the same sign. We may then by possibly negating $\bv_i$ assume that $b_{ij}, b_{ik}, \rho(\theta_1), \rho(\theta_2) \geq 0$. 

    For $\ell \in \{j, k\}$, let $\buperp{\ell}{0, i}$ be the component of $\bv_\ell$ parallel to $\bvperp{\ell}{0, i}$, so we have 
    \[
    \buperp{\ell}{0, i} = \sqrt{1 - b_\ell^2}\bv_\ell^\perp - \frac{b_{i\ell} - b_ib_\ell}{\sqrt{1 - b_i^2}}\bv_i^\perp.
    \] 
    It follows that 
    \begin{equation}
        \|\buperp{\ell}{0, i}\|^2 = 1 - b_\ell^2 - \frac{(b_{i\ell} - b_ib_\ell)^2}{1 - b_i^2}
    \end{equation}
    for $\ell \in \{j, k\}$, and
    \begin{equation}
        \buperp{j}{0, i} \cdot \buperp{k}{0, i} = (b_{jk} - b_jb_k) - \frac{(b_{ij} - b_ib_j)(b_{ik} - b_ib_k)}{1 - b_i^2}.
    \end{equation}
    For $\sigma \in \{-1, 1\}$ and $\ell \in \{j, k\}$, let us define 
    \begin{equation}
        p_\sigma = \frac{1 + \sigma b_i}{2}, \quad\text{ and }\quad q_{\sigma, \ell} = \frac{1 + \sigma b_i - \sigma b_\ell - b_{i\ell}}{2(1 + \sigma b_i)}.\footnote{These quantities correspond naturally to certain probabilities in the local distributions given by the configurations. We do not discuss this further since we did not define the local distributions.}
    \end{equation}
    By the triangle inequalities, we have $1 - b_i + b_\ell - b_{i\ell} \geq 0$, so $b_\ell \geq -1 + b_i + b_{i\ell} \geq -0.5 + b_i$. Since $|b_\ell| \leq 0.4$, we have $b_\ell \geq \max(-0.4, -0.5 + b_i)$. So we have 
    \begin{equation}
        q_{+1, \ell} = \frac{1 + b_i - b_\ell - b_{i\ell}}{2(1 + b_i)} \leq \frac{1 + b_i - \max(-0.4, -0.5 + b_i) - 0.5}{2(1 + b_i)}.
    \end{equation}
    When $b_i \geq 0.1$, the right hand side above is $\frac{1}{2(1 + b_i)} \leq \frac{5}{11}$. When $b_i \leq 0.1$, the right hand side above is $\frac{0.9+b_i}{2(1 + b_i)} = \frac{1}{2}-\frac{0.1}{2(1 + b_i)}\leq \frac{5}{11}$. So in either case we obtain $q_{+1, \ell} \leq \frac{5}{11}$. Similarly, for $q_{-1, \ell}$ we use the triangle inequality $1 + b_i - b_\ell - b_{i\ell} \geq 0$, from which we deduce $b_\ell \leq \min(0.4, 0.5+b_i)$, and therefore
    \begin{equation}
        q_{-1, \ell} = \frac{1 - b_i + b_\ell - b_{i\ell}}{2(1 - b_i)} \leq \frac{1 - b_i + \min(0.4, 0.5 + b_i) - 0.5}{2(1 - b_i)}.
    \end{equation}
    When $b_i \leq -0.1$, the right hand side above is $\frac{1}{2(1 - b_i)} \leq \frac{5}{11}$. When $b_i \geq -0.1$, the right hand side above is $\frac{0.9-b_i}{2(1 - b_i)} = \frac{1}{2}-\frac{0.1}{2(1 - b_i)}\leq \frac{5}{11}$. So in either case we obtain $q_{-1, \ell} \leq \frac{5}{11}$.

    Now let's relate the quantities $p_\sigma, q_{\sigma, \ell}$ to the SDP vectors. We have 
    \begin{align*}
        \sum_{\sigma \in \{-1, 1\}}p_\sigma q_{\sigma, \ell}(1 - q_{\sigma, \ell}) & = \sum_{\sigma \in \{-1, 1\}}\frac{1 + \sigma b_i}{2} \cdot \frac{(1 + \sigma b_i)^2 - (\sigma b_\ell + b_{i\ell})^2}{4(1 + \sigma b_i)^2} \\
        & = \frac{1}{4}\left(1 - \frac{(b_\ell + b_{i\ell})^2}{2(1 + b_i)} - \frac{(b_\ell - b_{i\ell})^2}{2(1 - b_i)}\right) \\ 
        & = \frac{1}{4} \|\buperp{\ell}{0, i}\|^2,
    \end{align*}
    as well as 
    \begin{align*}
        \sum_{\sigma \in \{-1, 1\}}p_\sigma q_{\sigma, j}q_{\sigma, k} & = \sum_{\sigma \in \{-1, 1\}}\frac{(1 + \sigma b_i - \sigma b_j - b_{ij})(1 + \sigma b_i - \sigma b_k - b_{ik})}{8(1 + \sigma b_i)} \\
        & = \sum_{\sigma \in \{-1, 1\}}\frac{1}{4}\left(\frac{1 + \sigma b_i - (\sigma b_j + b_{ij} + \sigma b_k + b_{ik})}{2} + \frac{(b_j + \sigma b_{ij})(b_k + \sigma b_{ik})}{2(1 + \sigma b_i)}\right) \\
        & = \frac{1}{4}\left(1 - b_{ij} - b_{ik} + \frac{(b_j + b_{ij})(b_k + b_{ik})}{2(1 + b_i)} + \frac{(b_j - b_{ij})(b_k - b_{ik})}{2(1 - b_i)}\right) \\
        & = \frac{1}{4}\left(1 - b_{ij} - b_{ik} + b_jb_k + \frac{(b_{ij} - b_ib_j)(b_{ik} - b_ib_k)}{1 - b_i^2}\right)\\
        & = \frac{1}{4}\left(1 - b_{ij} - b_{ik} + b_{jk} - \buperp{j}{0, i} \cdot \buperp{k}{0, i}\right) \\
        & \geq -\frac{1}{4}  \buperp{j}{0, i} \cdot \buperp{k}{0, i}.
    \end{align*}
    Here the last inequality is due to the triangle inequalities.

    Observe that since $q_{\sigma, \ell} \leq \alpha$ where $\alpha = \frac{5}{11}$ for every $\sigma \in \{-1, 1\}, \ell \in \{j, k\}$, we have
    \begin{equation}
        q_{\sigma, j}q_{\sigma, k} \leq \frac{\alpha}{1 - \alpha}\sqrt{q_{\sigma, j}(1 - q_{\sigma, j})q_{\sigma, k}(1 - q_{\sigma, k})}
    \end{equation}
    and therefore
    \begin{align*}
        \sum_{\sigma \in \{-1, 1\}} p_\sigma q_{\sigma, j}q_{\sigma, k} & \leq \frac{\alpha}{1 - \alpha} \cdot \sum_{\sigma \in \{-1, 1\}} p_\sigma \sqrt{q_{\sigma, j}(1 - q_{\sigma, j})q_{\sigma, k}(1 - q_{\sigma, k})} \\
        & \leq \frac{\alpha}{1 - \alpha} \cdot \sqrt{\left(\sum_{\sigma \in \{-1, 1\}} p_\sigma q_{\sigma, j}(1 - q_{\sigma, j})\right)}\sqrt{\left(\sum_{\sigma \in \{-1, 1\}} p_\sigma q_{\sigma, k}(1 - q_{\sigma, k})\right)},
    \end{align*}
    where the last inequality is due to Cauchy-Schwarz. Thus, we have 
    \begin{align*}
    \buperp{j}{0, i} \cdot \buperp{k}{0, i} & \geq -4 \sum_{\sigma \in \{-1, 1\}} p_\sigma q_{\sigma, j}q_{\sigma, k} \\
    & \geq -4 \cdot \frac{\alpha}{1 - \alpha} \cdot \sqrt{\left(\sum_{\sigma \in \{-1, 1\}} p_\sigma q_{\sigma, j}(1 - q_{\sigma, j})\right)}\sqrt{\left(\sum_{\sigma \in \{-1, 1\}} p_\sigma q_{\sigma, k}(1 - q_{\sigma, k})\right)} \\
    & = - \frac{\alpha}{1 - \alpha} \cdot \|\buperp{j}{0, i}\| \|\buperp{k}{0, i}\|
    \end{align*}
    By shifting the factors, we obtain
    \begin{equation}
        \bvperp{j}{0, i} \cdot \bvperp{k}{0, i} = \frac{\buperp{j}{0, i} \cdot \buperp{k}{0, i}}{\|\buperp{j}{0, i}\| \|\buperp{k}{0, i}\|} \geq -\frac{\alpha}{1 - \alpha} = -\frac{5}{6}.
    \end{equation}
    This completes the proof.

\bibliography{references}

@inproceedings{hsieh2023approximating,
  title={{Approximating Max-Cut on Bounded Degree Graphs: Tighter Analysis of the FKL Algorithm}},
  author={Hsieh, Jun-Ting and Kothari, Pravesh K.},
  booktitle={International Colloquium on Automata, Languages, and Programming (ICALP)},
  pages={77:1-77:7},
  year={2023}
}

@book{szeg1939orthogonal,
  title={Orthogonal {P}olynomials},
  author={Szegő, Gábor},
  series={Colloquium Publications},
  volume={23},
  year={1939},
  publisher={American Mathematical Society}
}

@book{o2014analysis,
  title={{Analysis of Boolean Functions}},
  author={O'Donnell, Ryan},
  year={2014},
  publisher={Cambridge University Press}
}

@inproceedings{brakensiek2024tight,
  title={{Tight Approximability of MAX 2-SAT and Relatives, Under UGC}},
  author={Brakensiek, Joshua and Huang, Neng and Zwick, Uri},
  booktitle={Symposium on Discrete Algorithms (SODA)},
  pages={1328--1344},
  year={2024}
}

@inproceedings{lewin2002improved,
  title={{Improved Rounding Techniques for the MAX 2-SAT and MAX DI-CUT Problems}},
  author={Lewin, Michael and Livnat, Dror and Zwick, Uri},
  booktitle={International Conference on Integer Programming and Combinatorial Optimization (IPCO)},
  pages={67--82},
  year={2002}
}

@inproceedings{brakensiek2023separating,
  title={{Separating MAX 2-AND, MAX DI-CUT and MAX CUT}},
  author={Brakensiek, Joshua and Huang, Neng and Potechin, Aaron and Zwick, Uri},
  booktitle={Symposium on Foundations of Computer Science (FOCS)},
  pages={234--252},
  year={2023}
}

@article{plackett1954reduction,
  title={{A Reduction Formula for Normal Multivariate Integrals}},
  author={Plackett, R. L.},
  journal={Biometrika},
  volume={41},
  number={3/4},
  pages={351--360},
  year={1954},
  publisher={JSTOR}
}

@inproceedings{austrin2007balanced,
  title={{Balanced MAX 2-SAT Might Not Be the Hardest}},
  author={Austrin, Per},
  booktitle={Symposium on Theory of Computing (STOC)},
  pages={189--197},
  year={2007}
}

@article{KKMO07,
	title={Optimal inapproximability results for MAX-CUT and other 2-variable CSPs?},
	author={Khot, Subhash and Kindler, Guy and Mossel, Elchanan and O’Donnell, Ryan},
	journal={SIAM Journal on Computing (SICOMP)},
	volume={37},
	number={1},
	pages={319--357},
	year={2007},
	publisher={SIAM}
}

@inproceedings{Rag08,
	author    = {Prasad Raghavendra},
	title     = {{Optimal Algorithms and Inapproximability Results for Every CSP?}},
	booktitle = {Symposium on Theory of Computing (STOC)},
	pages     = {245--254},
	year      = {2008},
	url       = {http://doi.acm.org/10.1145/1374376.1374414},
	doi       = {10.1145/1374376.1374414}
}

@inproceedings{Khot02a,
	author    = {Subhash Khot},
	title     = {{On the Power of Unique 2-Prover 1-Round Games}},
	booktitle = {Symposium on Theory of Computing (STOC)},
	pages     = {767--775},
	year      = {2002}
}

@article{Has01,
	title={Some optimal inapproximability results},
	author={H{\aa}stad, Johan},
	journal={Journal of the ACM (JACM)},
	volume={48},
	number={4},
	pages={798--859},
	year={2001},
	publisher={ACM New York, NY, USA}
}

@inproceedings{MM17CSP,
	title={Approximation algorithms for CSPs},
	author={Makarychev, Konstantin and Makarychev, Yury},
	booktitle={Dagstuhl Follow-Ups},
	volume={7},
	year={2017},
	organization={Schloss Dagstuhl-Leibniz-Zentrum fuer Informatik}
}

@inproceedings{BMORRSTVWW15,
	author       = {Boaz Barak and
	Ankur Moitra and
	Ryan O'Donnell and
	Prasad Raghavendra and
	Oded Regev and
	David Steurer and
	Luca Trevisan and
	Aravindan Vijayaraghavan and
	David Witmer and
	John Wright},
	title        = {{Beating the Random Assignment on Constraint Satisfaction Problems
	of Bounded Degree}},
	booktitle    = {Approximation, Randomization, and Combinatorial Optimization. Algorithms
	and Techniques, {APPROX/RANDOM}},
	pages        = {110--123},
	year         = {2015},
}

@misc{Floren16,
	title={Approximation of Max-Cut on Graphs of Bounded Degree},
	author={Flor{\'e}n, Mikael},
	year={2016}
}

@article{Khot05Guest,
	title={Guest column: inapproximability results via long code based PCPs},
	author={Khot, Subhash},
	journal={ACM SIGACT News},
	volume={36},
	number={2},
	pages={25--42},
	year={2005},
	publisher={ACM New York, NY, USA}
}

@inproceedings{KhotUGCSurvey,
	title={On the unique games conjecture (invited survey)},
	author={Khot, Subhash},
	booktitle={2010 IEEE 25th annual conference on computational complexity},
	pages={99--121},
	year={2010},
	organization={IEEE Computer Society}
}

@article{HLZ04,
	title={MAX CUT in cubic graphs},
	author={Halperin, Eran and Livnat, Dror and Zwick, Uri},
	journal={Journal of Algorithms},
	volume={53},
	number={2},
	pages={169--185},
	year={2004},
	publisher={Elsevier}
}

@inproceedings{BK99,
	title={On some tighter inapproximability results},
	author={Berman, Piotr and Karpinski, Marek},
	booktitle={Automata, Languages and Programming: 26th International Colloquium, ICALP’99 Prague, Czech Republic, July 11--15, 1999 Proceedings},
	pages={200--209},
	year={2002},
	organization={Springer}
}

@inproceedings{AKS09,
	title={Inapproximability of vertex cover and independent set in bounded degree graphs},
	author={Austrin, Per and Khot, Subhash and Safra, Muli},
	booktitle={2009 24th Annual IEEE Conference on Computational Complexity},
	pages={74--80},
	year={2009},
	organization={IEEE}
}

@article{Aus10,
	title={Towards sharp inapproximability for any 2-CSP},
	author={Austrin, Per},
	journal={SIAM Journal on Computing},
	volume={39},
	number={6},
	pages={2430--2463},
	year={2010},
	publisher={SIAM}
}

@inproceedings{GW94,
	author    = {Michel X. Goemans and
	David P. Williamson},
	title     = {.879-approximation algorithms for {MAX} {CUT} and {MAX} 2SAT},
	booktitle = {Proceedings of the Twenty-Sixth Annual {ACM} Symposium on Theory of
	Computing, 23-25 May 1994, Montr{\'{e}}al, Qu{\'{e}}bec,
	Canada},
	pages     = {422--431},
	publisher = {{ACM}},
	year      = {1994},
}

@article{FKL02,
  title={Improved approximation of Max-Cut on graphs of bounded degree},
  author={Feige, Uriel and Karpinski, Marek and Langberg, Michael},
  journal={Journal of Algorithms},
  volume={43},
  number={2},
  pages={201--219},
  year={2002},
  publisher={Elsevier}
}

@article{feige2002optimality,
	title={On the optimality of the random hyperplane rounding technique for MAX CUT},
	author={Feige, Uriel and Schechtman, Gideon},
	journal={Random Structures \& Algorithms},
	volume={20},
	number={3},
	pages={403--440},
	year={2002},
	publisher={Wiley Online Library}
}

@article{austrin2016better,
  title={Better balance by being biased: A 0.8776-approximation for max bisection},
  author={Austrin, Per and Benabbas, Siavosh and Georgiou, Konstantinos},
  journal={ACM Transactions on Algorithms (TALG)},
  volume={13},
  number={1},
  pages={1--27},
  year={2016},
  publisher={ACM New York, NY, USA}
}
\bibliographystyle{alpha}

\appendix

\section{Facts about Hermite Polynomials}\label{app:hermite}

In this appendix we collect a few facts about the probabilist's Hermite polynomials, to be used in Appendix~\ref{app:variance}. The $n$-th probabilist's Hermite polynomial, denoted by $\He_n$, is defined by
\begin{equation}
    \He_n(x) = (-1)^n \ee^{x^2/2} \frac{d^n}{dx^n}\ee^{-x^2/2}.
\end{equation}

For $\rho \in [-1, 1]$, we use $g_1, g_2 \sim_\rho \mathcal{N}(0, 1)$ to denote a sample of two standard Gaussian variables $g_1, g_2$ with covariance $\rho$. We use $\Phi_\rho$ and $\varphi_\rho$ to denote the c.d.f. and p.d.f. of this bivariate distribution respectively. The following is a standard fact for the probabilist's Hermite polynomials. 

\begin{lemma}\label{lem:hermite_orthogonality}
    Let $m, n \in \mathbb{N}$ and $\rho \in [-1, 1]$. Then we have
    \begin{equation}
        \E_{g_1, g_2 \sim_\rho \mathcal{N}(0, 1)}[\He_m(g_1)\He_n(g_2)] = \delta_{mn} \cdot m! \cdot \rho^m,
    \end{equation}
    where $\delta_{mn} = \begin{cases}
        1 & \text{if } m = n,\\
        0 & \text{otherwise}.
    \end{cases}$
\end{lemma}

The above lemma allows us to define an orthonormal basis using the probabilist's Hermite polynomials, and consequently we have:

\begin{lemma}\label{lem:hermite_inner_product}
    Let $f_1, f_2: \mathbb{R} \to \mathbb{R}$ be two square-integrable functions with respect to the standard Gaussian measure. Then we have
    \begin{equation}
        \E_{g_1, g_2 \sim_\rho \mathcal{N}(0, 1)}[f_1(g_1)f_2(g_2)] = \sum_{k = 0}^\infty \hat{f_1}(k)\hat{f_2}(k)\rho^k,
    \end{equation}
    where $\hat{f}(k) = \frac{1}{\sqrt{k!}}\E_{g \in \mathcal{N}(0, 1)}[f(g)\He_k(g)]$ is the $k$-th Hermite coefficient of $f$.
\end{lemma}

We refer the readers to Chapter 11.2 in~\cite{o2014analysis} for proofs of Lemma~\ref{lem:hermite_orthogonality} and Lemma~\ref{lem:hermite_inner_product}. 

\begin{lemma}\label{lem:indicator_hermite}
    Let $s_1, s_2 \in \mathbb{R}$ and $\rho \in [-1, 1]$. We have
    \begin{equation}
        \E_{g_1, g_2 \sim_\rho \mathcal{N}(0, 1)}[\bone[g_1 \geq s_1]\cdot\bone[g_2 \geq s_2]] = (1 - \Phi(s_1))(1 - \Phi(s_2)) + \sum_{k = 1}^\infty \frac{\He_{k-1}(s_1) \varphi(s_1)\He_{k-1}(s_2) \varphi(s_2)}{k!}\cdot\rho^k.
    \end{equation}
\end{lemma}
\begin{proof}
    Let us first compute the Hermite coefficients of $\bone[x \geq s]$. We have $\E_{g \sim \mathcal{N}(0, 1)}[\bone[g \geq s]\He_0(g)] = \E_{g \sim \mathcal{N}(0, 1)}[\bone[g \geq s]] = 1 - \Phi(s)$. For any $k \geq 1$, we have 
    \begin{equation}
        \E_{g \sim \mathcal{N}(0, 1)}[\bone[g \geq s] \cdot \He_k(g)] = \int_{s}^{\infty} \He_k(x) \varphi(x)\dd x = \He_{k-1}(s) \varphi(s),
    \end{equation}
    where in the last equality we used the fact that $\frac{\dd}{\dd x}\He_{k-1}(x)\varphi(x) = -\He_k(x)\varphi(x)$. The lemma now follows by Lemma~\ref{lem:hermite_inner_product}.
\end{proof}

\begin{lemma}\label{lem:plackett}
    Let $s_1, s_2 \in \mathbb{R}$ and $\rho \in (-1, 1)$.  We have
    \begin{equation}
        \frac{\partial}{\partial \rho}\E_{g_1, g_2 \sim_\rho \mathcal{N}(0, 1)}[\bone[g_1 \geq s_1]\cdot\bone[g_2 \geq s_2]] = \varphi_\rho(s_1, s_2).
    \end{equation}
\end{lemma}
\begin{proof}
    Follows directly from Plackett's identity~\cite{plackett1954reduction}.
\end{proof}

We use the following asymptotic estimate for Hermite polynomials.
\begin{lemma}[See e.g., Eq. 8.22.8 in~\cite{szeg1939orthogonal}]\label{lem:hermite_asymptotics}
    Let $s \in \mathbb{R}$. Then we have
    \begin{equation}
        \frac{\Gamma(k/2 + 1)}{\Gamma(k + 1)} \cdot \ee^{-\frac{s^2}{4}} \cdot \He_k(s)\cdot 2^{\frac{k}{2}} = \cos\left(\sqrt{k} \cdot s - k\pi/2\right) + O(k^{-1/2}).
    \end{equation} 
\end{lemma}

\begin{corollary}\label{cor:he_fact_asymptotics}
    Let $s \in \mathbb{R}$. Then we have
    \begin{equation}
        \frac{\He_{k-1}(s)}{\sqrt{k!}} = \ee^{s^2/4} \cdot \left(\frac{2}{\pi}\right)^{1/4} \cdot  k^{-3/4} \cdot \left(\cos\left(\sqrt{k-1} \cdot s- (k-1)\pi/2\right) + o(1)\right).
    \end{equation}
\end{corollary}
\begin{proof}
    By Lemma~\ref{lem:hermite_asymptotics} we have
    \begin{equation}\label{eq:he_factorial_expansion}
        \frac{\He_{k-1}(s)}{\sqrt{k!}} = \frac{\left(\cos\left(\sqrt{k - 1} \cdot s- (k-1)\pi/2\right) + O(k^{-1/2})\right) \cdot \Gamma(k)}{\sqrt{k!} \cdot 2^{(k-1)/2} \cdot \ee^{-s^2/4}\cdot \Gamma((k+1)/2)}.
    \end{equation}
    By Stirling's approximation we have $\Gamma(z) = \sqrt{2\pi}\ee^{-z}z^{z - 1/2} (1 + O(z^{-1}))$ for any $z > 0$, so we have
    \begin{align*}
        & \quad \frac{\Gamma(k)}{\sqrt{k!} \cdot 2^{(k-1)/2} \cdot \Gamma((k+1)/2)} \\
        & = \frac{\sqrt{2\pi}\ee^{-k}k^{k - 1/2} (1 + O(k^{-1}))}{\sqrt{\sqrt{2\pi}\ee^{-k-1}(k+1)^{k + 1/2} (1 + O(k^{-1}))} \cdot 2^{(k-1)/2} \cdot \sqrt{2\pi}\ee^{-(k+1)/2}((k+1)/2)^{k/2} (1 + O(k^{-1}))} \\
        & = \left(\frac{2}{\pi}\right)^{1/4} \cdot k^{-3/4} \cdot (1 + o(1)). 
    \end{align*}
    The corollary follows by plugging this back to~\eqref{eq:he_factorial_expansion}.
\end{proof}

\begin{corollary}\label{cor:hermite_tail_sum}
    Let $s \in \mathbb{R}$ and $m \in \mathbb{N}$. Then we have
    \begin{equation}
        \sum_{k = m}^\infty\frac{\He_{k-1}(s)^2}{k!} = \Theta(m^{-1/2}).
    \end{equation}
\end{corollary}
\begin{proof}
    By Corollary~\ref{cor:he_fact_asymptotics} we have
    \[ 
    \frac{\He_{k-1}(s)^2}{k!} = \Theta(k^{-3/2}) \cdot \cos^2(\sqrt{k-1} \cdot s- (k-1)\pi/2).
    \]
    Observe that the difference between the expressions inside cosine term tends to $\pi / 2$, so we have that when $k$ is sufficiently large
    \[
    \cos^2(\sqrt{k-1} \cdot s- (k-1)\pi/2) + \cos^2(\sqrt{k} \cdot s- (k)\pi/2) \geq 1/2.
    \]
    Hence, our corollary follows.
\end{proof}

\section{Variance of Sums of Gaussian Threshold Indicator Functions}\label{app:variance}

In this appendix, we analyze the variance of functions of the form 
\[\sum_{i = 1}^d w_i \bone\{g_i \geq s_i\},\] where $(g_1, \ldots, g_d)$ is a centered Gaussian vector with unit variances, $s_1, \ldots, s_d \in \mathbb{R}$ are some thresholds for the Gaussians, and $w_1, \ldots, w_d \geq 0$ are some nonzero weights. We prove a few different versions with different assumptions on the covariance structure or on the thresholds. They will be used to show that there is sufficient variance in a local neighborhood of a bounded-degree MAX 2-CSP instance for an algorithm to exploit, generalizing Lemma 7 from~\cite{hsieh2023approximating} which served a similar purpose for MAX CUT.

In the first version, there is no assumption on the covariance structure but the thresholds are assumed to be close to some nonzero real number.

\begin{lemma}\label{lem:variance_2sat}
    Let $(g_1, \ldots, g_d)$ be a centered Gaussian vector with unit variances. Let $s \in \mathbb{R}$ be nonzero, $\delta = |s| / 2$, $s_1, \ldots, s_d \in (s - \delta, s+ \delta)$, and $w_1, \ldots, w_d \in \mathbb{R}^{\geq 0}$, then we have
    \begin{equation}
        \Var\left[\sum_{i = 1}^d w_i \cdot \bone\{g_i \geq s_i\}\right] = \left(\sum_{i = 1}^d w_i\right)^2\cdot\Omega(d^{-1}),
    \end{equation}
    where the constant inside $\Omega$ depends only on $s$.
\end{lemma}
\begin{proof}

    We have
    \begin{align*}
        \Var\left[\sum_{i = 1}^d w_i \cdot \bone\{g_i \geq s_i\}\right] & = \sum_{i, j = 1}^d w_iw_j \cdot \Cov[\bone\{g_i \geq s_i\}, \bone\{g_j \geq s_j\}] \\
        & = \sum_{i, j = 1}^d w_iw_j \cdot \left(\E[\bone\{g_i \geq s_i\}\bone\{g_j \geq s_j\}] -\E[\bone\{g_i \geq s_i\}]\E[\bone\{g_j \geq s_j\}] \right) \\
        & = \sum_{i, j = 1}^d w_iw_j \cdot \sum_{k = 1}^\infty\frac{\He_{k-1}(s_i) \varphi(s_i)\He_{k-1}(s_j) \varphi(s_j)}{k!}\cdot\rho_{ij}^k.
    \end{align*}
    Here $\rho_{ij}$ is the covariance between $g_i$ and $g_j$. In the last equality we used Lemma~\ref{lem:indicator_hermite}. Observe that for every fixed $k \in \mathbb{N}$ we have
    \begin{equation}
        \sum_{i, j = 1}^d w_iw_j \cdot \frac{\He_{k-1}(s_i) \varphi(s_i)\He_{k-1}(s_j) \varphi(s_j)}{k!}\cdot\rho_{ij}^k \,\geq\, 0,
    \end{equation}
    since $(\rho_{ij}^k)_{1 \leq i, j \leq d}$ is a positive semidefinite matrix, so we have
    \begin{align*}
        \Var\left[\sum_{i = 1}^d w_i \cdot \bone\{g_i \geq s_i\}\right] 
        & \geq \sum_{i, j = 1}^d w_iw_j \cdot \frac{\He_{1}(s_i) \varphi(s_i)\He_{1}(s_j) \varphi(s_j)}{2}\cdot\rho_{ij}^2 \\
        & = \sum_{i, j = 1}^d w_iw_j \cdot \frac{s_i \varphi(s_i)s_j \varphi(s_j)}{2}\cdot\rho_{ij}^2,
    \end{align*}
    where we used $\He_1(x) = x$. By our assumption, we have $\frac{s_i \varphi(s_i)s_j \varphi(s_j)}{2} > 0$, so we may ignore the non-diagonal terms and obtain
    \begin{equation}
    \Var\left[\sum_{i = 1}^d w_i \cdot \bone\{g_i \geq s_i\}\right] \geq \sum_{i = 1}^d w_i^2 \cdot  \frac{s_i^2 \varphi(s_i)^2}{2} \geq c_s \sum_{i = 1}^d w_i^2 \geq c_s \cdot \frac{1}{d} \left(\sum_{i = 1}^d w_i\right)^2,
    \end{equation}
    for some $c_s > 0$ depending only on $s$. 
\end{proof}

In the second version we assume that the covariances are somewhat bounded away from 1 in absolute value, and that the thresholds are also absolutely bounded.

\begin{lemma}\label{lem:variance_2}
    Fix some constants $B > 0$. Let $(g_1, \ldots, g_d)$ be a centered Gaussian vector with unit variances and pairwise covariances in $\left[-1 + \frac{1}{d^2\log^2d}, 1 - \frac{1}{d^2\log^2d}\right]$. Let $s_1, \ldots, s_d \in [-B, B]$ and $w_1, \ldots, w_d \in \mathbb{R}^{\geq 0}$, then we have
    \begin{equation}
        \Var\left[\sum_{i = 1}^d w_i \cdot \bone\{g_i \geq s_i\}\right] = \left(\sum_{i = 1}^d w_i\right)^2\cdot\Omega(d^{-2} \log^{-3/2}d),
    \end{equation}
    where the constant inside $\Omega$ depends only on $B$.
\end{lemma}
\begin{proof}
    As in the proof of Lemma~\ref{lem:variance_2sat}, for any $m \geq 1$ we have
    \begin{equation}\label{eq:sum_truncated}
        \Var\left[\sum_{i = 1}^d w_i \cdot \bone\{g_i \geq s_i\}\right] \geq \sum_{i, j = 1}^d w_iw_j \cdot \sum_{k = m}^\infty\frac{\He_{k-1}(s_i) \varphi(s_i)\He_{k-1}(s_j) \varphi(s_j)}{k!}\cdot\rho_{ij}^k.
    \end{equation}
    We break the sum on the right hand side of Equation~\eqref{eq:sum_truncated} into two parts: diagonal terms and non-diagonal terms. For diagonal terms we have by Corollary~\ref{cor:hermite_tail_sum} and Cauchy-Schwarz:
    \begin{equation}
        \sum_{i = 1}^d w_i^2 \cdot \sum_{k = m}^\infty\frac{(\He_{k-1}(s_i) \varphi(s_i))^2}{k!} \geq \frac{\left(\sum_{i = 1}^d w_i\right)^2}{d} \cdot \Omega(m^{-1/2}).
    \end{equation}

    For the off-diagonal terms, we have
    \begin{align*}
        & \,\left|\sum_{i \neq j} w_iw_j\cdot \sum_{k = m}^\infty\frac{\He_{k-1}(s_i) \varphi(s_i)\He_{k-1}(s_j) \varphi(s_j)}{k!}\cdot\rho_{ij}^k \right|\\
        \leq \, & \,\left(\sum_{i = 1}^d w_i\right)^2\cdot \sum_{k = m}^\infty\left|\frac{\He_{k-1}(s_i) \varphi(s_i)\He_{k-1}(s_j) \varphi(s_j)}{k!}\right|\cdot\left(1 - \frac{1}{d^2\log^2d}\right)^m \\
        = \, & \,\left(\sum_{i = 1}^d w_i\right)^2\cdot O(m^{-1/2})\cdot\left(1 - \frac{1}{d^2\log^2d}\right)^m
    \end{align*}
    where in the last step we again used Corollary~\ref{cor:hermite_tail_sum}. Now we take $m = c d^2 \log^3d$ for some sufficiently large constant $c > 0$ depending only on $B$, and it follows that the sum on the right hand side of Equation~\eqref{eq:sum_truncated} is $\left(\sum_{i = 1}^d w_i\right)^2\cdot\Omega(d^{-2} \log^{-3/2}d)$.
\end{proof}

In the last version, we assume that the covariances are bounded away from -1 by a constant, and again the thresholds are bounded in absolute value.

\begin{lemma}\label{lem:variance_main}
    Fix some constants $\eta, B > 0$. Let $(g_1, \ldots, g_d)$ be a centered Gaussian vector with unit variances and pairwise covariances at least $-1 + \eta$. Let $s_1, \ldots, s_d \in [-B, B]$ and $w_1, \ldots, w_d \in \mathbb{R}^{\geq 0}$, then we have
    \begin{equation}
        \Var\left[\sum_{i = 1}^d w_i \cdot \bone\{g_i \geq s_i\}\right] = \left(\sum_{i = 1}^d w_i\right)^2\cdot\Omega(d^{-2} \log^{-3/2}d),
    \end{equation}
    where the constant inside $\Omega$ depends only on $\eta$ and $B$.
\end{lemma}

\begin{proof}
    Consider a new centered Gaussian vector $(g_1', \ldots, g_d')$ such that 
    \[
    g_i' = \sqrt{1 - d^{-2}\log^{-2}d}\cdot g_i + (d^{-1}\log^{-1}d) \cdot g_i'',
    \]
    where $g_1'', \ldots, g_d''$ are standard Gaussian variables independent from all other variables. Then $g_i'$ is also a standard Gaussian and $\Cov[g_i', g_j'] = (1 - d^{-2}\log^{-2}d)\cdot \Cov[g_i, g_j]$ for every distinct $i,j$. We may therefore apply Lemma~\ref{lem:variance_2} to $g_1', \ldots, g_d'$ and obtain
    \begin{equation}\label{eq:g_iprime}
        \Var\left[\sum_{i = 1}^d w_i \cdot \bone\{g_i' \geq s_i\}\right] = \left(\sum_{i = 1}^d w_i\right)^2\cdot\Omega(d^{-2} \log^{-3/2}d).
    \end{equation}
    We now compare $\Var\left[\sum_{i = 1}^d w_i \cdot \bone\{g_i' \geq s_i\}\right]$ and $\Var\left[\sum_{i = 1}^d w_i \cdot \bone\{g_i \geq s_i\}\right]$. For every pair $i, j$ such that $\Cov[g_i, g_j] \geq 0$, we have
    \begin{equation}
        \Cov[\bone\{g_i \geq s_i\}, \bone\{g_j \geq s_j\}] \geq \Cov[\bone\{g_i' \geq s_i\}, \bone\{g_j' \geq s_j\}]
    \end{equation}
    by Lemma~\ref{lem:plackett} since $\Cov[g_i, g_j] \geq \Cov[g_i'
    , g_j']$. For every pair $i, j$ such that $\Cov[g_i, g_j] < 0$, again by Lemma~\ref{lem:plackett} we have
    \begin{align*}
        & \Cov[\bone\{g_i' \geq s_i\}, \bone\{g_j' \geq s_j\}] - \Cov[\bone\{g_i \geq s_i\}, \bone\{g_j \geq s_j\}] \\
        \leq\, & (\Cov[g_i', g_j']-\Cov[g_i, g_j])  \cdot c(\eta, B)\\
        \leq\,& d^{-2}\log^{-2}d \cdot c(\eta, B),
    \end{align*}
    where $c(\eta, B) = \max_{\rho \in [-1 + \eta, 0], t_1, t_2 \in [-B, B]} \varphi_\rho(t_1, t_2)$. Thus, we have
    \begin{equation}
          \Var\left[\sum_{i = 1}^d w_i \cdot \bone\{g_i \geq s_i\}\right]- \Var\left[\sum_{i = 1}^d w_i \cdot \bone\{g_i' \geq s_i\}\right] 
         \geq \,  -\sum_{i, j = 1}^d w_i w_j \cdot d^{-2}\log^{-2}d \cdot c(\eta, B).
    \end{equation}
    Combined with Equation~\eqref{eq:g_iprime} we obtain $\Var\left[\sum_{i = 1}^d w_i \cdot \bone\{g_i \geq s_i\}\right] = \left(\sum_{i = 1}^d w_i\right)^2\cdot\Omega(d^{-2} \log^{-3/2}d)$ as desired.
\end{proof}

\end{document}